%% file: main.tex
\documentclass[journal]{IEEEtran}

\usepackage[caption=false,font=footnotesize,labelfont=sf,textfont=sf]{subfig}

\usepackage{textcomp}
\usepackage{stfloats}
\usepackage{url}
\usepackage{verbatim}
\usepackage{graphicx}
\def\BibTeX{{\rm B\kern-.05em{\sc i\kern-.025em b}\kern-.08em
    T\kern-.1667em\lower.7ex\hbox{E}\kern-.125emX}}
\usepackage{balance}

\usepackage{cite}
\usepackage{amsmath,amssymb,amsfonts}
\usepackage{algorithm}
\usepackage{algorithmicx}
\usepackage{algpseudocode}
\usepackage{setspace}
\usepackage{bm}
\usepackage{graphicx}
\usepackage{textcomp}
\usepackage{xcolor}
\usepackage{tikz}
\usetikzlibrary{quantikz2}
\usepackage{float}

\usepackage[T1]{fontenc}

\usepackage[hidelinks]{hyperref}

\usepackage{makecell}
\usepackage{booktabs}
\usepackage{ulem} 
\usepackage{array} 

\usepackage{amsthm}

\usepackage{graphbox} 
\usepackage{adjustbox}
\usepackage{balance}
\usepackage{multirow}
\usepackage[capitalize]{cleveref}
\usepackage{footnote}
\usepackage{bm}
\usepackage{threeparttable}

\usepackage{lipsum}
\usepackage[table]{xcolor}
\usepackage{colortbl}
\usepackage{nicematrix}
\usepackage[usestackEOL]{stackengine}

\usepackage{footmisc}
\usepackage{enumitem}

\newtheorem{theo}{Theorem}
\newtheorem{lem}{Lemma}
\newtheorem{cor}{Corollary}
\theoremstyle{remark}

\theoremstyle{definition}
\newtheorem{defin}{Definition}

\newtheorem*{remark}{Remark}

\newtheorem*{prob}{Research Problem}
\newtheorem*{main}{Main idea}

\newcommand{\eqdef}{\stackrel{\triangle}{=}}

\begin{document}

\title{
Quantum Routing Beyond Pathfinding: \\ Multipartite Entanglement Complementation
{}
}

\author{Si-Yi Chen, Angela~Sara~Cacciapuoti 
 and Marcello~Caleffi
    \thanks{A preliminary conference version of this work is \cite{CheCacCal-25-QCE}.} 
    \thanks{The authors are with the \href{www.quantuminternet.it}{www.QuantumInternet.it} research group, University of Naples Federico II, Naples, 80125 Italy. \textit{Corresponding author: Angela Sara Cacciapuoti (e-mail: angelasara.cacciapuoti@unina.it).}}
    \thanks{\textcolor{black}{This work has been funded by the European Union under Horizon Europe ERC-CoG grant QNattyNet, n.101169850.} Views and opinions expressed are however those of the author(s) only and do not necessarily reflect those of the European Union or the European Research Council Executive Agency. Neither the European Union nor the granting authority can be held responsible for them.}
    }

\maketitle

\begin{abstract}
Conventional quantum routing operates under the entrenched assumption that pathfinding is a prerequisite for routing. 
This classical-inspired routing model imposes a restricting design option, which prevents scaling the quantumness to the network functioning. In this paper, we proposed a novel entanglement-driven routing framework that exploits multipartite entanglement complementation for enabling simultaneous 1-hop connectivity among all non-adjacent source-destination pairs. This changes the notion of ``remoteness'' in the \textcolor{black}{entanglement graph}, activated by entanglement. We extend this framework to inter-domain quantum networks and design a polynomial-time algorithm. Such an algorithm allows to select and parallelize multiple requests, bypassing NP-complete path discovery. Performance analysis shows the proposed routing strategy achieves up to $60\%$ hop reduction, 
with the algorithm enabling efficient parallelism and strong scalability in inter-domain quantum networks.
\end{abstract}


\section{Introduction}
\label{sec:1}

\IEEEPARstart{E}{nabling} end-to-end entanglement between remote network nodes is often regarded as a primary task in quantum network applications\cite{CalCac-26,CacIllCal-23, ConMalWin-24,LiNarAel-24,GuLiYu-24}. Research in this field\cite{AbaCubMai-25,ShiQia-20,AliChe-22,Cal-17,LiuLi-25,GyoImr-18-1,LiLiLiu-21,HerPomBeu-22,SatIshNag-16} has extensively drawn inspiration from classical routing principles, by discovering optimized Source-Destination (S-D) paths based on selected link metrics, such as hop-count or fidelity, and then extending entanglement along the individuated path via quantum repeaters. Nevertheless, this classical-inspired paradigm carries over inherent limitations, including excessive quantum memory demands for  \textcolor{black}{\textit{Routing-Qubit Footprints} (RQF, see Tab.~\ref{tab:01})}\cite{CalAmoFer-24, rfc9340}. Indeed, the entrenched assumption that path selection is a prerequisite imposes restrictive design constraints, particularly in scaling the quantumness to the network functioning \cite{CalCac-26}. Addressing these limitations requires moving beyond path-dependent paradigms through an entanglement-driven routing framework.

\begin{figure}[t]
    \centering

    \subfloat[Graph $G$]{%
        \includegraphics[width=0.485\linewidth]{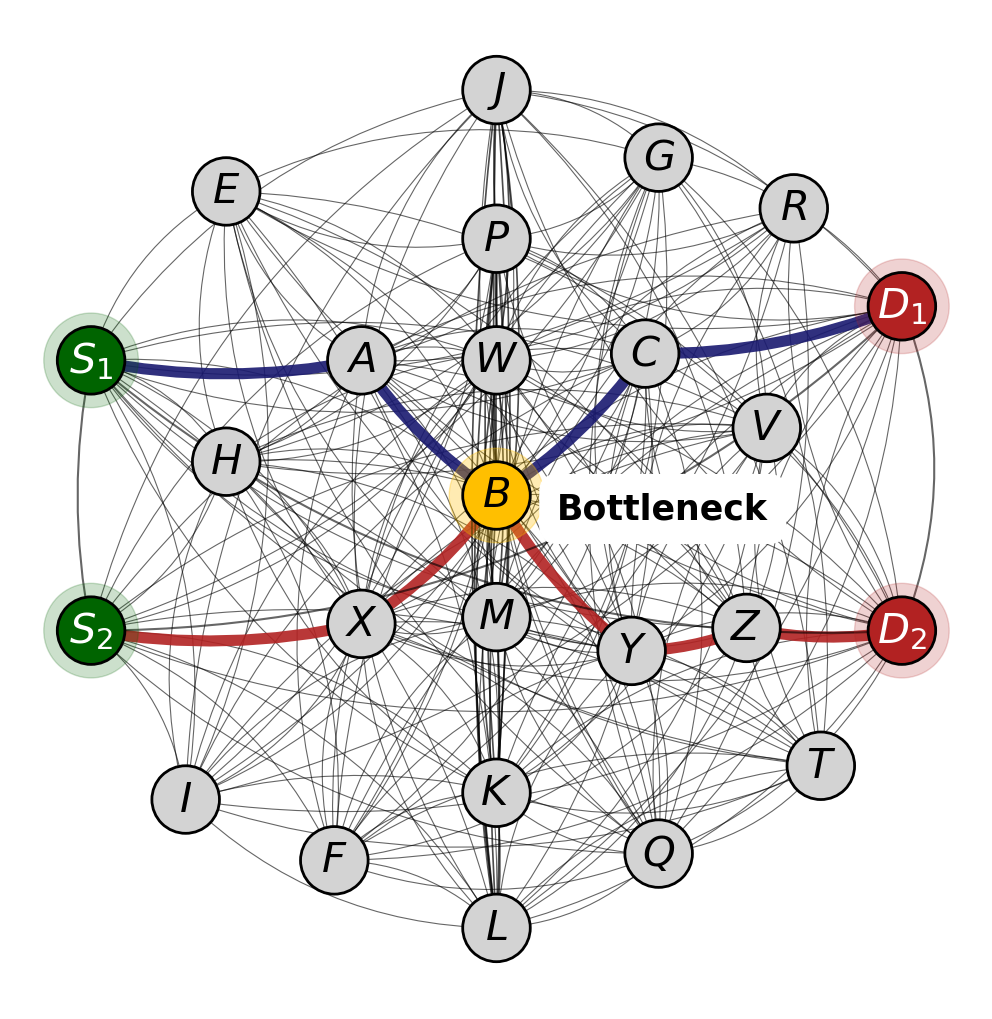}%
        \label{fig:01.a}
    }
    \hfill
    \subfloat[Complement Graph $\bar{G}$]{%
        \includegraphics[width=0.485\linewidth]{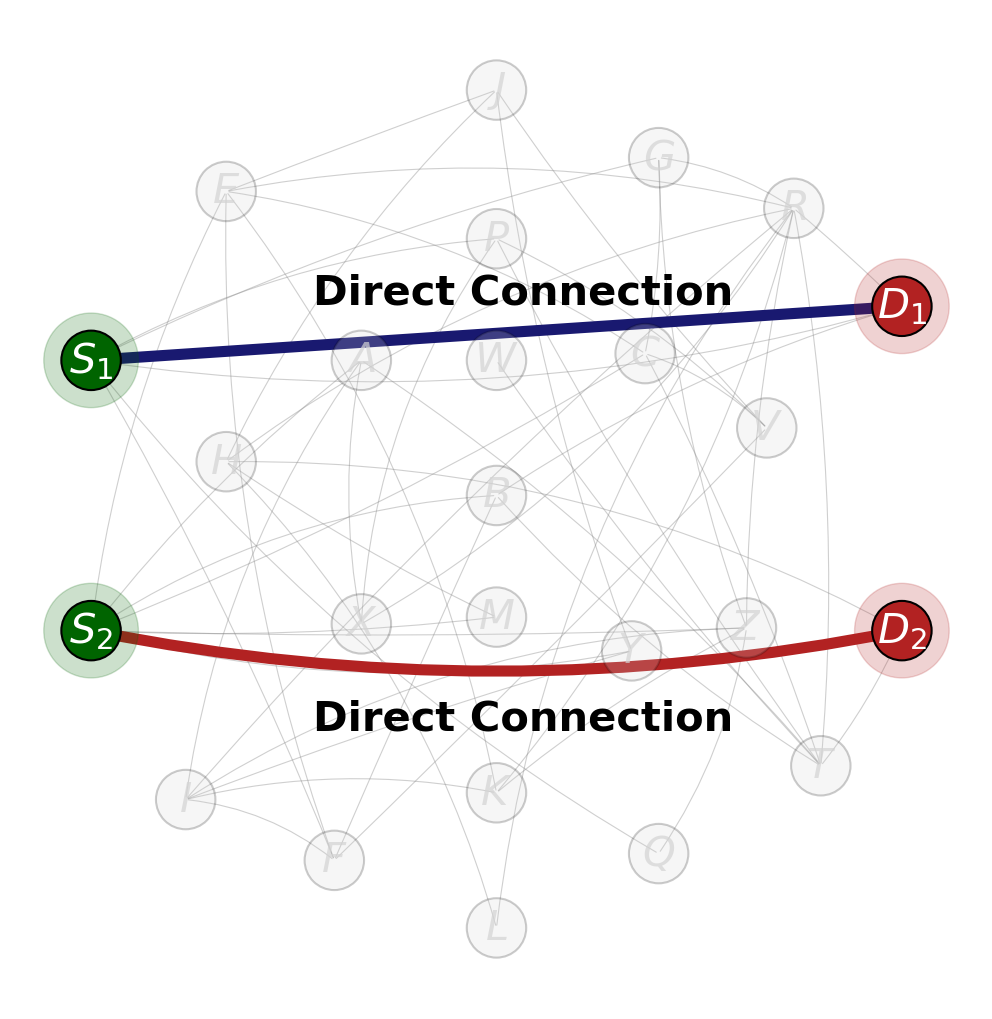}%
        \label{fig:01.b}
    }

    \caption{\textcolor{black}{Bottleneck Paths in  Graph $G$ vs Direct Parallel Connections in  Complement Graph $\bar{G}$.}}
    \label{fig:01}
    \vspace{0.5ex}
    \hrulefill
\end{figure}

\renewcommand{\arraystretch}{1.4}
\begin{table*}[t]
\centering
\caption{\textcolor{black}{Summary of Global Notation}}
\label{tab:02}

\fontsize{8pt}{8pt}\selectfont 
\begin{tabular}{|c|p{0.36\linewidth}|c|c|p{0.36\linewidth}|}
\hline
\hline
\multicolumn{1}{|c|}{\textbf{ \textcolor{black}{Acronyms} }} & \multicolumn{1}{c|}{\textbf{ \textcolor{black}{Full Meaning} }} & &
\multicolumn{1}{c|}{\textbf{ \textcolor{black}{Acronyms} }} & \multicolumn{1}{c|}{\textbf{ \textcolor{black}{Full Meaning} }}   
\\
\hline
    \textcolor{black}{ \textit{MEC} } & \textcolor{black}{ Multipartite Entanglement Complementation } &  &    \textcolor{black}{ \textit{BSM} } & \textcolor{black}{Bell-State Measurement}\\

    \textcolor{black}{ \textit{CQR} } & \textcolor{black}{ Conventional Quantum Routing } &  &    \textcolor{black}{ \textit{DP} } & \textcolor{black}{Dynamic Parallel Pairs algorithm}\\

    \textcolor{black}{ \textit{QNet} } & \textcolor{black}{ Quantum Network (e.g., QLAN) } &  &    \textcolor{black}{ \textit{RQF} } & \textcolor{black}{ Routing-Qubit Footprint}\\

    \textcolor{black}{ \textit{S-D} pair } & \textcolor{black}{ Source-Destination pair } &  &    \textcolor{black}{ \textit{ARQF} } & \textcolor{black}{Aggregate Routing-Qubit Footprint}\\

    \textcolor{black}{ \textit{NDP} } & \textcolor{black}{Node-Disjoint Paths problem} &  &    \textcolor{black}{ \textit{} } & \textcolor{black}{}\\
\hline
\hline
\multicolumn{1}{|c|}{\textbf{ \textcolor{black}{Symbols} }} & \multicolumn{1}{c|}{\textbf{ \textcolor{black}{Definitions} }} & &
\multicolumn{1}{c|}{\textbf{ \textcolor{black}{Symbols} }} & \multicolumn{1}{c|}{\textbf{ \textcolor{black}{Definitions} }}   
\\ 
\hline

    \textcolor{black}{$V_a$} & \textcolor{black}{Set of nodes in QNet $a$.} &  &  \textcolor{black}{$R$} & \textcolor{black}{A batch of remote EPR requests. Each request is a S-D pair.} \\

    \textcolor{black}{$v_{a.i}$} & \textcolor{black}{The $i$-th node in QNet $V_a$. Represents a quantum node with one qubit.} &  & \textcolor{black}{$\lambda$} & \textcolor{black}{Interval between request rounds.	Used in throughput derivation.} \\

    \textcolor{black}{$E_{a.i,b.j}$} & \textcolor{black}{Inter-link between $v_{a.i}$ with $v_{b.j}$ (Def.~\ref{def:x01}). Only exists across different QNets ($a \neq b$). } &  & \textcolor{black}{$\bar{h}$} &  \textcolor{black}{Average hop-count per request in CQR. Used in throughput and resource comparison.} \\

    \textcolor{black}{$E$} & \textcolor{black}{Set of established inter-links.} &  & \textcolor{black}{$\bar{r}$} & \textcolor{black}{Average number of parallel requests per cycle.} \\

    \textcolor{black}{$\bar{E}$} &  \textcolor{black}{Complement inter-link set (Def.~\ref{def:x04}). Represents all possible inter-links not present in original graph.} &  & \textcolor{black}{$\rho$}  & \textcolor{black}{Number of cycles to complete all requests. Depends on parallelism.}    \\

    \textcolor{black}{$C$} & \textcolor{black}{Control-node system, fully connected with control nodes.} &  &  \textcolor{black}{$t_i$} & \textcolor{black}{The arrival time of the $i$-th  batch of requests.}  \\

    \textcolor{black}{$c_a$} & \textcolor{black}{Control node associated with QNet $V_a$. Connected to all nodes in $V_a$.} &  & \textcolor{black}{$ft_i$} & \textcolor{black}{The completion time of the $i$-th batch of requests.}  \\

    \textcolor{black}{$k$} & \textcolor{black}{Number of QNets.} &  & \textcolor{black}{$p$}   & \textcolor{black}{Edge probability controlling the density of synthetic inter-QNet graphs.}  \\

    \textcolor{black}{$k'$}  & \textcolor{black}{Number of control nodes, defined as $k'= k+(k\mod2)$.} &  & \textcolor{black}{$\chi$} & \textcolor{black}{Numbers of intermediate nodes to complete all requests in CQR.} \\

    \textcolor{black}{${N}_a(v_i)$}  & \textcolor{black}{Set of nodes in other QNets adjacent to $v_i$.} &  & \textcolor{black}{$n^M, n^B$} & \textcolor{black}{Number of cycles can be executed within the interval by MEC / CQR.} \\
    
    \textcolor{black}{$\overline{N}_a(v_i)$}  & \textcolor{black}{Set of nodes in other QNets \textbf{not} adjacent to $v_{a.i}$. Complement of ${N}_a(v_i)$.} &  &  \textcolor{black}{$F^M, F^B$} & \textcolor{black}{Throughput of MEC / CQR. Number of successful EPR pairs per time shot.}  \\

    \textcolor{black}{$G_{n_1,\dots,n_k}$} & \textcolor{black}{Inter-QNet graph (Def.~\ref{def:x02}), collection of Inter-Links $E$ and all QNet nodes. }
    &  &  \textcolor{black}{$T_{P}^M, T_{R}^M$} & \textcolor{black}{Preparation / Routing time in MEC.  $T_{P}^M$: Controlled inter-QNet generation and distribution time.  $T_{R}^M$: Local Pauli-measurement performing time.} \\

    \textcolor{black}{$\bar{G}_{n_1,\dots,n_k}$} & \textcolor{black}{Complement Inter-QNet graph (Def.~\ref{def:x05}), collection of Complement inter-Links $\bar{E}$ and all QNet nodes. }
    &  &  \textcolor{black}{$T_{P}^B, T_{R}^B$} & \textcolor{black}{Preparation / Routing time in CQR. $T_{P}^B$: Time for EPR generation and distribution along the $\bar{h}$-hop path. $T_{R}^B$: Bell-State Measurement performing time.} \\

    \textcolor{black}{$CG_{n_1,\dots,n_k}$} & \textcolor{black}{Controlled inter-QNet graph (Def.~\ref{def:x03}). Includes control-node system and all QNets.} &  & \textcolor{black}{$Q_{R}^B$} & \textcolor{black}{The total number of entanglement-resource qubits required network-wide during the routing phase to serve a batch of S--D requests by CQR.}  \\

     &  &  & \textcolor{black}{$Q^{M,\mathrm{pro}}_{R}$,$Q^{M,\mathrm{ond}}_{R}$} & \textcolor{black}{The total number of entanglement-resource qubits required network-wide during the routing phase to serve a batch of S--D requests by proactive / on-demand MEC.}  \\

\hline
\hline
\end{tabular}
\end{table*}

Specifically, by sharing multipartite entanglement resources among the network nodes, it is possible to activate dynamic \textit{\textcolor{black}{entanglement graphs}}\footnote{\label{ft:02} An artificial link is pictorially visualized as an edge between two nodes connected in \textcolor{black}{the entanglement graph. In entanglement-based networks, we advocate for the precise use of the term “graph” to describe the dynamic arrangement of artificial links, and reserve “topology” for invariant structural properties (e.g., hierarchy, containment) -- a distinction rooted in~\cite{Day-07}.} 
\textcolor{black}{For relevant Entanglement-based network  technical jargon, see \cite{CheCacCal-26}.}} built upon the physical graph, which can be adaptively manipulated to establish end-to-end entanglement. This, in turn, allows to track  evolving communication needs \cite{PirDur-19,CheIllCac-25,HahPapEis-19,CheIllCac-24-QCE,WenJosSte-18,LiXueLi-23,MazCalCac-25,MorDur-24,PomHerBai-21,HumKalMor-18,BenHajVan-24,ManPat-23}.    

In this context, solving the quantum routing problem goes beyond simply finding ``optimal'' physical paths according to a selected routing metric. It involves manipulating the \textcolor{black}{entanglement graph} to establish end-to-end entanglement \cite{CalCac-26}. Thus, the routing task shifts from classical next-hop tracking to proactively managing the entanglement-activated overlay graph, ensuring end-to-end entanglement distribution \cite{CalCac-26} and enabling dynamic reconfiguration of the \textcolor{black}{entanglement graph} for source-destination interconnection.

However, this shift introduces new challenges. At scale, dynamically reconfigured artificial links\footref{ft:02} resemble the tangled wires of traditional headphones, where each new connection risks "knotting" the network, demanding complex resource coordination. Consequently, a fundamental rethinking of quantum routing is essential, by eliminating multi-node path complexity through direct (1-hop) artificial links between communicating nodes.
Remarkably, the concept of graph complement~\cite{BonMur-76} aligns precisely with the above need. In fact, it inverts the connectivity relationships among nodes, by making non-adjacent nodes in the original \textcolor{black}{entanglement graph} adjacent in the complement, as illustrated in Fig.~\ref{fig:01}. This, as better analyzed in Sec.\ref{sec:2.2}, reorients the focus from pathfinding to \textcolor{black}{entanglement graph} engineering. 
\textcolor{black}{Tab.~\ref{tab:02} summarizes the global notation used throughout this paper. For clarity, algorithm-specific notations are listed separately in Sec.~\ref{sec:4.2} (see Tab.~\ref{tab:notations}).}
 
Building on this idea, we introduce the approach of \textit{multipartite entanglement complementation} (MEC), in which the connectivity relationships, activated by multipartite entanglement among the network nodes, are fully inverted.

More into details, in this work, we propose a MEC strategy for \textit{Inter-domain networking}, 
which is able to simultaneously transform all source–destination pairs that are remote in the original \textcolor{black}{entanglement graph} into neighbors. Remarkably, this is achieved at minimal cost, requiring only a single \textcolor{black}{qubit} per node. As better clarified later, such a constraint, \textcolor{black}{namely a \textit{Routing-Qubit Footprint} (RQF) of one qubit per node}, fundamentally breaks Bell-state–based approaches. Differently, the proposed MEC strategy sustains 1-hop inter-domain entanglement across all non-adjacent S-D pairs,  even under this extreme resource limitation. 
Furthermore, we design a polynomial-time algorithm  that enables the simultaneous selection of multiple S-D pairs for end-to-end entanglement establishment. Our contributions are summarized as follows.
\begin{itemize}
    \item We propose a MEC strategy for inter-domain networking, aimed at establishing end-to-end entanglement across nodes belonging to different networks. We further provide a comparative analysis against conventional quantum routing approaches.
    \item We design a polynomial-time algorithm  for selecting all feasible S-D pair candidates that satisfy the parallel condition, by allowing autonomous pair selection while circumventing the NP-hard node-disjoint paths problem.
    \item We validate the effectiveness of the MEC strategy and the scalability of the proposed algorithm, through extensive simulations, supported by a detailed performance and implementation analysis.
\end{itemize}

The remaining part of the paper is organized as follows. Section~\ref{sec:2} introduces the MEC strategy and compares it with conventional quantum routing. Section~\ref{sec:3} presents preliminaries on inter-domain networks, outlines the research problem, and describes our main idea. Section~\ref{sec:4} applies MEC to the proposed Controlled Inter-QNet network and designs a dynamic parallel-pairs algorithm for autonomous source–destination selection. Section~\ref{sec:5} validates the proposal, by using both real-world and synthetic datasets. Finally, Section~\ref{sec:6} concludes the paper \textcolor{black}{and discusses key aspects related to the proposed approach}.

\section{Preliminaries}
\label{sec:2}

\subsection{Conventional Quantum Routing (CQR)}
\label{sec:2.1}

\begin{figure}[t]
    \centering    
    \includegraphics[width=0.5\textwidth]{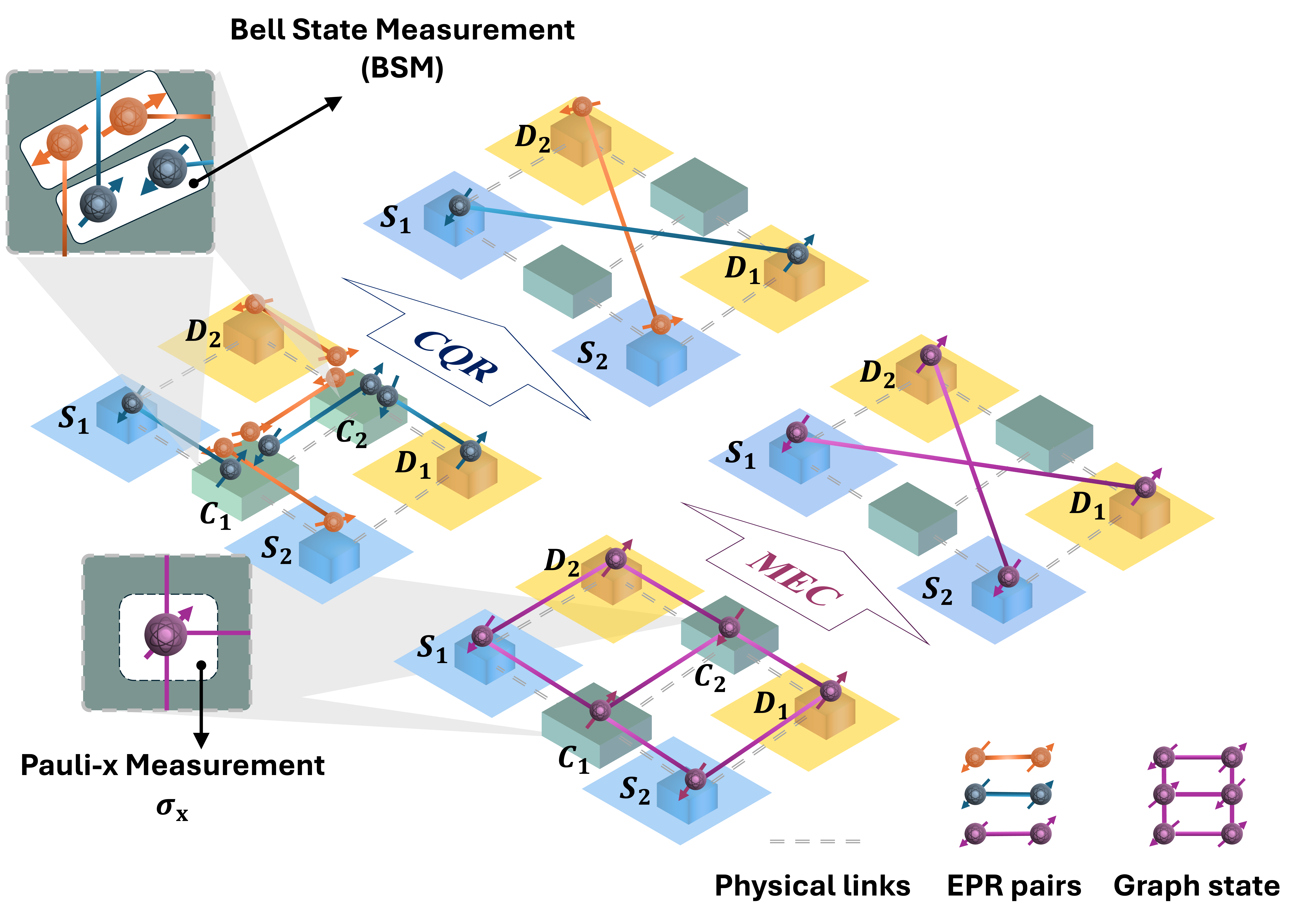}
    \caption{Conventional Quantum Routing (CQR) VS Multipartite Entanglement Complementation (MEC) in butterfly network. In CQR, multiple requests ($S_1, D_1$), ($S_2, D_2$) are managed via time-multiplexing, requiring intermediate nodes $C_1, C_2$ to allocate sufficient \textcolor{black}{qubits print} (4 qubits per node for 2 requests) to avoid bottlenecks. In contrast, \textcolor{black}{MEC leverages a 6-qubit graph state} to serve multiple requests simultaneously, with $C_1, C_2$ acting as control nodes to switch initial artificial links $(S_1, D_2), (S_2, D_1)$ into the complement counterparts $(S_1, D_1), (S_2, D_2)$, \textcolor{black}{achieving the routing-qubits footprint} to 1 qubit per node for 2 requests.} 
    \label{fig:02}    
    \vspace{0.5ex}
    \hrulefill
\end{figure}

\noindent As discussed in the Introduction, 
usually, in CQR, once entanglement resources (typically EPR pairs) are generated, intermediate nodes extend entanglement along multi-hop S–D paths by performing entanglement swapping through Bell-State Measurements (BSMs). A critical factor in this process is the  \textcolor{black}{\textit{Routing-Qubit Footprint} (RQF)} \cite{CalAmoFer-24,rfc9340}. Indeed, \textcolor{black}{releasing the constraint of RQF} significantly improves the rate of establishing end-to-end entanglement and, hence, it enhances the degree of communication parallelism in CQR. Conversely, when the \textcolor{black}{RQF} is limited, establishing end-to-end entanglement in parallel across multiple links becomes infeasible. As shown in Fig.~\ref{fig:02}, if intermediate nodes $C_1$, $C_2$ are restricted to only two qubits each, requests can only be processed sequentially. While routing optimization can partially mitigate this limitation, the problem ultimately reduces to the NP-complete node-disjoint paths (NDP) problem~\cite{Jen-95}, underscoring the inherent complexity of scalable CQR under strict resource constraints.

\subsection{Multipartite Entanglement Complementation (MEC)}
\label{sec:2.2}

\begin{table*}[!t]
\caption{Conventional Quantum Routing (CQR) VS Multipartite Entanglement Complementation (MEC)}
\label{tab:01}
\centering
\footnotesize
\setlength{\tabcolsep}{16pt} 
\renewcommand{\arraystretch}{1.75}
    \begin{threeparttable}[t]
    \begin{tabular}{c|c|c} 
        \toprule
        \toprule
        & \textbf{Conventional Quantum Routing} & \textbf{Multipartite Entanglement  Complementation}\\
        \cline{1-3}
        Key Operation & Entanglement Swapping & Graph complementation \\
        \cline{1-3}
        Traffic & Transit-dependent Paths  & Direct End-to-End Delivery \\
        \cline{1-3}
        Entanglement Resource & Typical: EPR pairs & Graph State \\
        \cline{1-3}
        \cellcolor{blue!5} S-D pair Selection & \cellcolor{blue!5} \textit{Pre-determined} & \cellcolor{blue!5} \textit{Dynamic} \\
        \cline{1-3}
        \cellcolor{blue!5} Delay & \cellcolor{blue!5} \textit{Multi-hop} & \cellcolor{blue!5} \textit{Single-hop} \\
        \cline{1-3}
        \cellcolor{blue!5}  Parallelism    & \cellcolor{blue!5} \tnote{a} \, S-D \textit{Paths} Parallel   & \cellcolor{blue!5} \tnote{b} \, S-D \textit{Pairs} Parallel 
        \\        
        \cline{1-3}
        \cellcolor{blue!5} \tnote{c} \, \textcolor{black}{Routing-Qubit Footprint (RQF)} 
        & \cellcolor{blue!5}  $>1$   & \cellcolor{blue!5} $=1$ (even for parallel multiple requests)  \\
        \bottomrule 
    \end{tabular}
    \begin{tablenotes}
        \item[a] S-D Paths Parallel: A set of S-D paths where no two paths share any node. This corresponds to the node-disjoint paths (NDP) problem, which seeks such path sets and is NP-complete~\cite{Jen-95}.
        \item[b] S-D Pairs Parallel: A set of S-D pairs whose endpoints satisfy  Disjoint Endpoint and Neighborhood Exclusion (see Lem.~\ref{lem:x02}). 
        \item[c] 
        \textcolor{black}{Routing-Qubit Footprint (RQF): The minimum number of entanglement-resource qubits per node that must be simultaneously available as inputs to routing operations during the routing stage to serve a routing instance. For multiple concurrent S--D requests, RQF denotes the per-node qubit footprint required to process them in parallel under the considered routing model. For MEC, ${\rm RQF}=1$ per node is sufficient for any set of S--D requests satisfying the parallel-pair condition (see Lem.~\ref{lem:x02}).}.
    \end{tablenotes}
    \end{threeparttable}
\end{table*}

\noindent As aforementioned, we exploit multipartite states, precisely graph states, to establish end-to-end entanglement across nodes belonging to different networks, thereby reducing the routing problem to a matter of graph manipulation~\cite{CheIllCac-25,PirDur-19,WenJosSte-18,CheIllCac-24-QCE,LiXueLi-23,MazCalCac-25,MorDur-24,PomHerBai-21,HumKalMor-18,BenHajVan-24,HahPapEis-19,ManPat-23}. The tool, we exploit to build the proposed Multipartite Entanglement Complementation (MEC) strategy\footnote{\label{ft:distribution}In our strategy, we assume that the entanglement resources are proactively generated and distributed in advance, with pathfinding required during the initial distribution stage, by relying on the physical network graph. \textcolor{black}{We acknowledge that generating and distributing multipartite entanglement generally entails overhead, as discussed in Section~\ref{sec:6}.}}, is graph complementation, which allows  to complement the \textcolor{black}{entanglement graph} associated with a multipartite graph state via a controlled procedure involving Pauli-$x$ measurements. To illustrate this correspondence, we present a schematic diagram comparing operations in the graph and graph state domains in Fig.~\ref{fig:03}. 

Differently from conventional approaches that primarily address pairwise entanglement, the MEC orients the focus toward the ``global properties'' of multiparty entanglement. As indicated in Fig.~\ref{fig:03}, MEC operationally admits two complementary versions: the original state $\ket{G}$, which encodes direct connectivity in the \textcolor{black}{entanglement graph}, and its complement, $\ket{\bar{G}}$, where artificial links directly connect nodes that are remote in the original state $\ket{G}$. A control-node system can dynamically switch between these two versions (see Sec.~\ref{sec:4.1}), by ensuring that for any S–D pair, either $\ket{G}$ or $\ket{\bar{G}}$ provides a direct (1-hop) artificial link. Thus, the traditional distinction between neighboring and remote node pairs\footnote{\label{ft:remote}Throughout this paper, the notions of ``\textit{remoteness}'' and its counterpart ``\textit{neighbor / adjacency}'' refer not to the physical proximity of network nodes, but to ``\textit{entangled proximity}'', that is, proximity within the \textcolor{black}{entanglement graph} associated with the graph $G$ \cite{CheIllCac-25}.} is eliminated, by enabling single-hop connectivity.

It is worthwhile to note that prior works exploited local complementation to modify individual or partial node neighborhoods, as in \cite{CheIllCac-25,MazCalCac-25,PirDur-19,HahPapEis-19,CheIllCac-24-QCE,AdcMorDah-20}. However, these contributions remain limited to local optimizations. In contrast, the proposed MEC  constructs the entire complementary graph, by elevating  complementation from a local optimization tool to a global-network reconfiguration mechanism. 

We summarize the key differences between  MEC and CQR strategies in Table~\ref{tab:01}. As, it will become clearer in the following sections, the proposed MEC strategy \textcolor{black}{achieves the RQF} to just one per node.  
This efficiency arises from leveraging a shared graph-state resource, which enables massive parallelism in request processing via complementation. Crucially, this parallelism operates at the level of S-D pairs rather than node-disjoint paths, thereby avoiding the contention issues inherent to CQR. For instance, in a butterfly network (see Fig.~\ref{fig:02}), a canonical inter-domain graph, CQR requires up to 4 qubits per node to support 2 parallel requests, whereas MEC achieves the same parallelism with only 1 qubit per node\footnote{For a detailed discussion of the crucial role played by \textcolor{black}{Routing-Qubit Footprint} and the importance of limiting their number to reduce complexity, the reader is referred to \cite{CacPelIll-26, CalAmoFer-22, rfc9340}.}. Furthermore, unlike CQR, MEC inherently supports dynamic S-D pair selection at the same routing-qubit footprint cost. In MEC, the graph state is prepared as a general-purpose resource, independent of any specific source-destination pair. Upon request arrival, complementation  dynamically configures this resource to establish direct links for the intended S-D pairs, by eliminating the multi-hop delays inherent in static, pre-determined path-based CQR.

\begin{figure}[t]
    \centering    \includegraphics[width=0.475\textwidth]{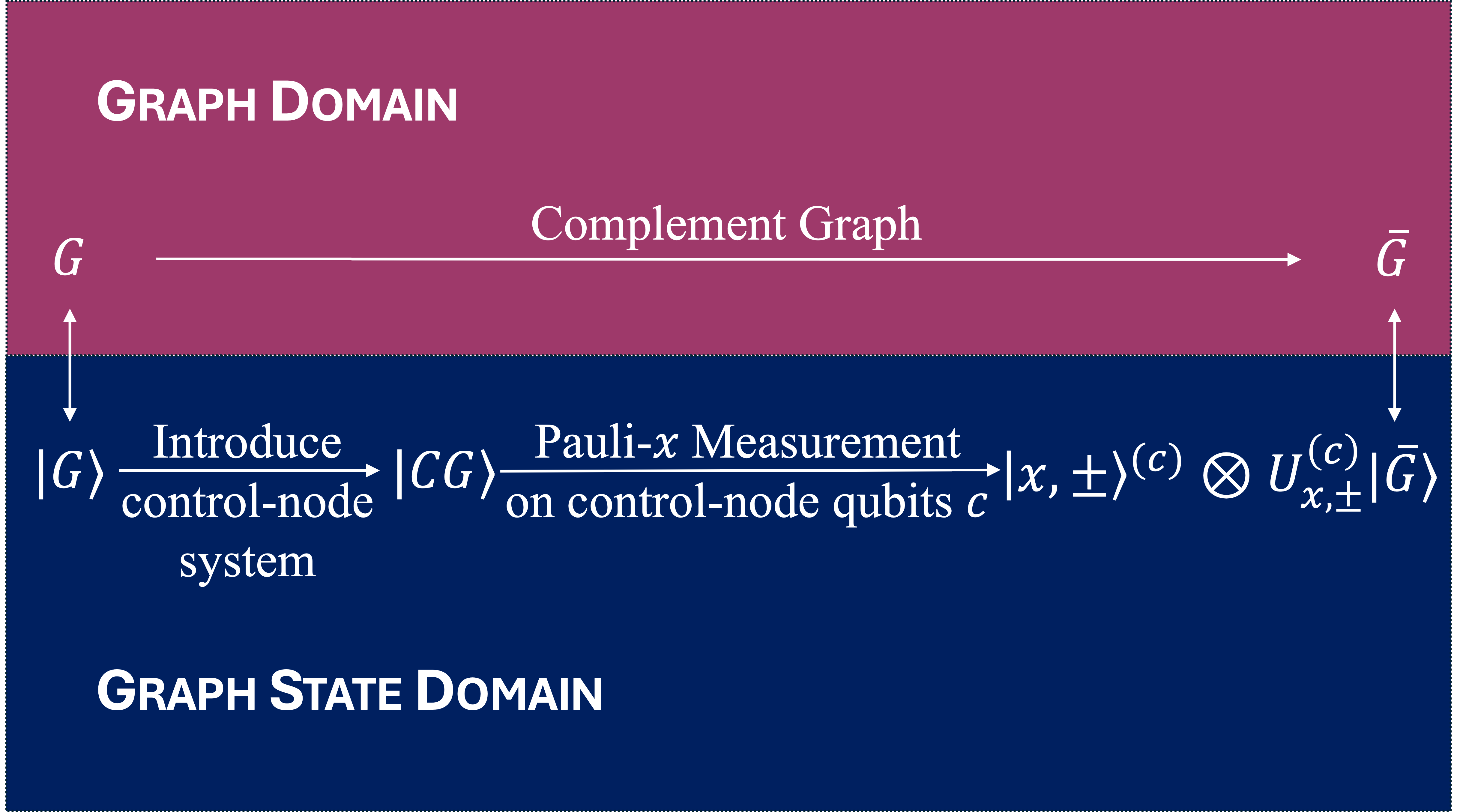}
    \caption{Schematic diagram of the correspondence between graph complement in \textit{graph} domain and Multipartite Entanglement Complementation in \textit{graph state} domain, i.e, of the mapping between introducing control-node system then $x$-measurement through Pauli operators on graph states and transformations of the associated graphs.}
    \label{fig:03}    
    \vspace{0.5ex}
    \hrulefill
\end{figure}

\section{Research Problem}
\label{sec:3}

As aforementioned, we address the quantum routing problem under the most stringent resource constraint, \textcolor{black}{namely a Routing-Qubit Footprint (RQF) of one qubit per node}. In this context, we focus on routing across multiple Quantum networks \textcolor{black}{(QNets)}, which serve as fundamental building blocks of the considered architecture.

To formally define the research problem, some preliminary definitions are needed. We assume a graph state\footref{ft:distribution} $\ket{G}$ distributed across $k$ different QNets, where each qubit resides at a distinct network node. The graph state is associated with a graph $G$ \cite{IllCalMan-22,CheCacCal-26,MazCalCac-25}.

\begin{defin}[\textbf{Inter-Links}]
    \label{def:x01}
    Given $k$ QNets, with node sets $V_1, \dots, V_k$, an inter-link is defined as an artificial link $E_{a.i, b.j}$ connecting a node $v_i \in V_a$ in QNet $V_a$ to a node $v_j \in V_b$ in a different QNet $V_b$, with $a\neq b$. Formally:
    \begin{equation}
        \label{eq:def:x01.1}
        E_{a.i,b.j}  \eqdef (v_i \in V_a , v_j \in V_b), 
    \end{equation}
    with $a, b \in \{1,\dots,k\}$, and $a\neq b$. For notational convenience, we also denote the node $v_i \in V_a$ as $v_{a.i}$.
\end{defin}
We generalize the concept of vertex-adjacency\footref{ft:remote} set, by defining $N_a(v_i)$ as the set of nodes in other QNets,
$\bigcup_{b=1, b\neq a}^k V_b$, that are connected to $v_i \in V_a$ via inter-links. For the sake of notation simplicity, in the following, we use the symbol $E$ to denote the set of established inter-links among $k$ QNets. 

Correspondingly, for non-adjacency\footref{ft:remote} set, we define $\overline{N}_a(v_i)$ as the set of nodes in other QNets,
$\bigcup_{b=1, b\neq a}^k V_b$, that are \textbf{not} connected to $v_i \in V_a$ via inter-links:
\begin{equation}
    \label{eq:def:x04.2}
    \overline{N}_a(v_i) \eqdef \big\{ v_j \in \bigcup_{b=1, b\neq a}^k V_b : (v_i,v_j ) \not\in E \big\}.
\end{equation}

\begin{defin}[\textbf{Complement Inter-Links}]
    \label{def:x04}
    Given $k$ QNets, with node sets $V_1, \dots, V_k$, the complement inter-links of a node $v_i \in V_a$ in QNet $V_a$ is defined as a set of all possible inter-links $\bar{E}(v_i)$ between $v_i$ and its non-adjacency node set $\overline{N}_a(v_i)$. 
    Formally:
    \begin{equation}
        \label{eq:def:x04}
        \bar{E}(v_i) \eqdef \big\{ (v_i, v_j) \bigm| v_j \in \overline{N}_a(v_i) \big\},
    \end{equation}
    with $v_i \in V_a$, $a \in \{1, \dots, k\}$, and $\overline{N}_a(v_i)$ given by \eqref{eq:def:x04.2}.
\end{defin}

Similarly, we use the symbol $\bar{E}$ to denote the union set of complement Inter-links among $k$ QNets.

\begin{defin}[\textbf{Inter-QNet}]
    \label{def:x02}
    Given $k$ QNets $V_1, \dots, V_k$, if there exists a set of inter-links that connects them into a single connected inter-domain network, we say that these $k$ QNets are \textit{interactive}, and the resulting network constituted by all the QNets nodes with their inter-links is called an Inter-QNet. The Inter-QNet is described by the graph state $\ket{ G_{n_1, \dots, n_k} }$, with associated graph $G_{n_1, \dots, n_k}$ given by:
    \begin{equation}
        \label{eq:def:x02}
        G_{n_1,\dots, n_k}  \eqdef (V_1, \dots, V_k, E), 
    \end{equation}
    with $|V_i|=n_i$ and inter-link set $E$, defined above.
\end{defin}

\begin{defin}[\textbf{Complement Inter-QNet}]
    \label{def:x05}
    Given an Inter-QNet $\ket{ G_{n_1,\dots, n_k} }$  with associated graph $G_{n_1,\dots, n_k}=(V_1,\dots, V_k, E)$, its complement Inter-QNet $\ket{ \bar{G}{n_1,\dots, n_k} }$ corresponds to the graph $\bar{G}_{n_1,\dots, n_k}$ consisting of the same QNets node sets $V_1, \dots, V_k$, while the inter-links set is given by the complement inter-links set $\bar{E}$. Formally:
    \begin{equation}
        \label{eq:def:x05}
        \bar{G}_{n_1,\dots, n_k} \eqdef  \big ( V_1,\dots, V_k, \bar{E} \big), 
    \end{equation}
    with $|V_i|=n_i$ and complement inter-links set $\bar{E}$, defined above.
\end{defin}

\begin{figure*}[t]
    \centering    \includegraphics[width=\textwidth]{Figures/MainIdea.png}
    \caption{Pictorial illustration of research problem and main idea. Consider a Controlled Inter-QNet $\ket{CG_{3,2,2,2}}$ with corresponding graph $CG_{3,2,2,2}$ shown in Fig.~\ref{fig:04}a\textsuperscript{5}. It consists of a control-node system (Fig.~\ref{fig:04}b) and the original Inter-QNet network (Fig.~\ref{fig:04}c), comprising four QNets represented by blue, red, green and yellow nodes. A set $R$ of remote EPR requests is listed in Tab.~I. Rather than searching for individual paths, by Pauli-$x$ measurement on all control nodes, we transform the original Inter-QNet network into its complement one (Fig.~\ref{fig:04}d), where all the requested pairs in $R$ become directly connected. Subsequently, the Dynamic Parallel Pairs algorithm (Alg.~\ref{alg:01}, Sec.~\ref{sec:4.2}) identifies parallel pairs for each requested pairs in $R$, which are then autonomously selected as shown in Fig.~\ref{fig:04}e. The final dynamic parallel pairs are mapped to Tab.~II.A-C and summarized in Tab.~II. Specifically, as shown in Tab.~II.A-C, each round lists all parallel pairs candidates for the current target, with the next one autonomously chosen (marked in red).}
    \label{fig:04}
    \vspace{0.5ex}
    \hrulefill
\end{figure*}
\footnotetext[5]{\textcolor{black}{Fig.~\ref{fig:04}.a depicts the number of logical qubits -- i.e., the vertices of the target controlled inter-QNet state -- while the physical qubit overhead associated to the generation process depends on the underlying hardware. Please refer to Sec.~\ref{sec:6} for detail.}}

Accordingly, the Inter-QNet concept emphasizes the connectivity activated by the shared entanglement across distinct QNets. We augment the Inter-QNet with a control-node system, resulting in a Controlled Inter-QNet, similar to \cite{CalCac-26}. This design choice introduces a distinction between two types of network nodes: the original QNet nodes, referred to simply as \textit{nodes}, and the nodes of the control system, referred to as \textit{control nodes}. Alternative design choices are possible: for example electing a node within a QNet as a control node, according to a specific optimization criterion, as done in \cite{CheIllCac-25,MazCalCac-25}. Importantly, such variations do not affect the analysis; they only require adjusting the definitions of the node sets $V_1, \dots, V_k$ to reflect the different node classes. We do not explore these design alternatives further, as they lie outside the scope of this paper and do not affect the key insights provided. Furthermore, by introducing the control system as follows, we are providing a more flexible framework, which allows researchers to adopt different design choices, depending on the specific application. 

In the following, we denote with  $k'$ the number of control nodes and as proved in the App.~\ref{app:lem:x01}, $k'$ should be set as $k'= k+(k\mod2)$, where $k$ denotes the number of QNets. The control system has the features summarized as follows:
\begin{enumerate}
    \item Control-node system: The control nodes form a fully connected control-layer $C$ on top of the Inter-QNet network.
    \item Node-Control Interface: Each QNet $V_a$ is associated with a control node $c_a$, which is entangled-connected with the nodes $v_i \in V_a$.
\end{enumerate}
Accordingly, and as formalized in Def.~\ref{def:x03}, the control system and the QNets are represented by a graph state denoted as $\ket{ CG_{n_1, \dots, n_k} }$. Fig.~\ref{fig:04}a shows a controlled Inter-QNet $\ket{CG_{3,2,2,2}}$, with its control-node system in Fig.\ref{fig:04}b and the original Inter-QNet network in Fig.~\ref{fig:04}c. The above structure is formally captured in the following definition.

\begin{defin}[\textbf{Controlled Inter-QNet}]
    \label{def:x03}
    Given an Inter-QNet $\ket{ G_{n_1,\dots, n_k} }$ with corresponding graph
    $G_{n_1,\dots, n_k}=(V_1,\dots, V_k, E)$, a controlled Inter-QNet $\ket{ CG_{n_1, \dots, n_k} }$ augments it with a control-node system, constituted by $k'$ nodes. Each control node $c_i$ is entangled-connected to all the nodes in its associated QNet $V_i$, and all control nodes $c_j$ are mutually interconnected, forming a fully connected control layer. The corresponding graph $CG_{n_1, \dots, n_k}$, associated to the resulting overall graph state, is:
    \begin{align}
        \label{eq:def:x03}
        & CG_{n_1,\dots, n_k} =
        \\ \nonumber
        = & \big(
        \big( \bigcup_{i=1}^k V_i \big) \bigcup \big( \bigcup_{j=1}^{k'} c_j \big),  \\ \nonumber
         & E \cup 
         \big( \bigcup_{ \substack{ 1\leq i,j \leq k'\\ i\neq j} } (c_i,c_j) \big) \cup
         \big( \bigcup_{a=1}^k \big(
         \bigcup_{ \substack{ v_{a.i} \in V_a } } (v_{a.i}, c_a)  \big) \big)
        \big),
    \end{align}
with $k'= k+(k\mod2)$, and $\bigcup_{i=1}^k V_i, \,E$ defined in~\eqref{eq:def:x02}.
\end{defin}

Building on our inter-domain network framework, we formally define the research problem.

\begin{prob}
Let us consider an inter-domain quantum network, constituted by $k$ interactive QNets with a control-node system, in which a multipartite graph state is distributed. Nodes in any QNet can request the establishment of EPR pairs with arbitrary non-adjacent\footref{ft:remote}  nodes located in different QNets. These requests are aggregated by  control-nodes into a set $R$ of remote EPR requests, where each request $r \in R$ represents a source-destination (S-D) pair. The objective is to design a quantum routing scheme that:
    \begin{itemize}
        \item[i)] Establishes EPR pairs for all requested S-D pairs in $R$;
        \item[ii)] Autonomously selects multiple S-D pairs in $R$ to communicate in parallel.
    \end{itemize}
\end{prob}

\begin{main}
We leverage the controlled Inter-QNet $\ket{CG_{n_1,\dots,n_k}}$ entangled resource for solving the routing problem through two stages: (1) constructing the complement Inter-QNet (see Sec.~\ref{sec:4.1}) and (2) enabling parallel communication opportunities for  multiple S-D pairs (see Sec.~\ref{sec:4.2}). This approach differs fundamentally from conventional quantum routing:
\begin{itemize}
    \item \textit{\textcolor{black}{Entanglement graph} engineering}: We reconfigure the network graph into its complementary form, by converting all requested S–D pairs into direct neighbors for end-to-end entanglement establishment. (see Sec.~\ref{sec:4.1}).
    \item \textit{Autonomous parallelism}: The proposed polynomial-time algorithm identifies compatible parallel pairs from $R$, by enabling autonomous selection of communication targets (see Sec.~\ref{sec:4.2}).
\end{itemize}
\end{main}
For clarity, Fig.~\ref{fig:04} provides a pictorial illustration of the research problem. 

\section{Beyond Pathfinding: Multipartite Entanglement Complementation}
\label{sec:4}

\subsection{Controlled Inter-QNet Complementation}
\label{sec:4.1}

\noindent As outlined in the research problem, rather than individually searching paths for each S-D pair, we aim to simultaneously establish end-to-end entanglement for all S-D pairs in $R$. This is achieved via the Multipartite Entanglement Complementation method, as introduced in the following Lemma~\ref{lem:x01}.

\begin{lem}[\textbf{Multipartite Entanglement Complementation}]
\label{lem:x01}
    Let us consider a given Controlled Inter-QNet $\ket{CG_{n_1,\dots, n_k}}$ with corresponding graph $CG_{n_1,\dots, n_k}$ in \eqref{eq:def:x03}. By performing Pauli-$x$ measurements $\sigma_x$ on all control nodes, the  nodes are projected onto the complement Inter-QNet $\ket{\bar{G}_{n_1,\dots, n_k}}$, where  nodes become directly connected to their remote counterparts in different QNets. Formally,
    \begin{equation}
        \label{eq:lem:x01.2}
        \ket{CG_{n_1, \dots, n_k}} \xrightarrow[\text{on } \cup_{i=1}^{k'} c_i]{\sigma_x}
        \bigotimes\limits_{i=1}^{k'} \ket{x,\pm}^{(c_i)} \otimes U^{(\cup_{i=1}^{k'} c_i)}_{x,\pm} \ket{\bar{G}_{n_1, \dots, n_k}},
    \end{equation}
    where the corresponding graph $\bar{G}_{n_1,\dots, n_k}$ is given in~\eqref{eq:def:x05}. 
    \begin{IEEEproof}
        Please refer to App.~\ref{app:lem:x01}. 
    \end{IEEEproof}
\end{lem}

\begin{cor}
\label{cor:x01}
The controlled inter-QNet state $\ket{CG_{n_1,\dots,n_k}}$ in Lemma~\ref{lem:x01} alternatively yields the original inter-QNet state $\ket{G_{n_1,\dots,n_k}}$ via Pauli-$z$ measurements on all control nodes. 
\end{cor}

\begin{remark}
    Lem.~\ref{lem:x01} and Cor.~\ref{cor:x01} demonstrate that we can dynamically switch between the original Inter-QNet $\ket{G}$ and the complement Inter-QNet $\ket{\bar{G}}$ configurations through appropriate measurements on the control-node system, enabling reconfigurable quantum networking. This capability essentially reflects \textit{controlled Inter-QNet complementation}.
\end{remark}

By Lem.~\ref{lem:x01}, the immediate establishment of EPR pairs for any S-D pair in Objective i) is achieved. This is because the resulting complement graph state $\ket{\bar{G}_{n_1,\dots,n_k}}$ possesses inter-links that correspond precisely to the desired S-D pairs specified in the request set $R$. Measuring extra qubits of the $\ket{\bar{G}_{n_1,\dots,n_k}}$ state therefore directly projects onto the EPR pairs connecting the desired S-D pairs.

\subsection{Dynamic Parallel Pairs Algorithm}
\label{sec:4.2}

\noindent Given \textcolor{black}{a RQF of one qubit per node} and no resource contention or serialization, \textit{parallel pairs} refer to a set of source-destination pairs that can simultaneously establish EPR pairs in a single, conflict-free round. In classical routing, achieving parallelism reduces to the NP-complete Node-Disjoint Paths (NDP)\cite{Jen-95}. To circumvent this complexity, we restate sufficient conditions for parallelism (see Lem.~\ref{lem:x02}), tailored to our proposed MEC routing scheme. Building on this, we propose a dynamic parallel pair algorithm (see Alg.~\ref{alg:01}) that operates directly on the complement network to identify multiple  parallel pairs. This approach fundamentally eliminates the reliance on path optimization and dynamically generates a list of parallel-pairs candidates for all S-D requests in polynomial time.

\begin{lem}[\textbf{Conditions for Parallel Pairs}]
    \label{lem:x02}
    Given an Inter-QNet $\ket{G_{n_1,\dots,n_k}}$ with associated graph $G_{n_1,\dots,n_k} = (V_1, \dots, V_k, E)$, two inter-links $E_{a.i,b.j}, E_{c.m,d.n} \in E$ form parallel pairs if they satisfy: 
    \begin{enumerate}[label=(\roman*)]
        \item Disjoint Endpoints: These links share no common vertices;
        \item Neighborhood Exclusion: Neither link endpoints are neighbors of another link endpoints.
    \end{enumerate}
    i.e., if:
    \begin{align}
        \label{eq:def:x06}
       & \{v_{a.i}, v_{b.j}\} \cap \{v_{c.m}, v_{d.n}\} = \emptyset \; \wedge \\ \nonumber
       \wedge \;
       &   \{v_{a.i}, v_{b.j}\} \cap \big(N(v_{c.m}) \cup N(v_{d.n})\big) = \emptyset.
    \end{align}
    \begin{IEEEproof}
        Performing Pauli-$z$ measurement on all nodes except $v_{a.i}, v_{b.j}, v_{c.m}$ and $v_{d.n}$ projects the state $\ket{G_{n_1,\dots,n_k}}$ onto a subgraph containing only the two inter-links, i.e., $E_{a.i,b.j}$ and $E_{c.m,d.n}$. Under the given conditions, these inter-links are both Disjoint Endpoints and  Neighborhood Exclusion, ensuring no entanglement conflict. Thus two independent EPR pairs simultaneously established: one along $E_{a.i,b.j}$ and the other along $E_{c.m,d.n}$, forming parallel pairs.         
    \end{IEEEproof}
\end{lem}

We say that inter-links $E_{a.i,b.j}$ and $E_{c.m,d.n}$ are \textit{parallel-pair candidates} of each other if they satisfy the conditions of Lem.~\ref{lem:x02} and can thus form parallel pairs. To generalize this pairwise compatibility to larger sets, we define:

\begin{defin}[\textbf{Parallel-Pairable}]
    \label{def:x06}
    Given an Inter-QNet $\ket{ G_{n_1,\dots, n_k} }$ with corresponding graph $G_{n_1,\dots, n_k}=(V_1, \dots, V_k, E)$, a set of inter-links $\mathcal{E} \subseteq E$ is \textit{parallel-pairable} if every pair of inter-links in $\mathcal{E}$ satisfies Lemma~\ref{lem:x02}.
\end{defin}

\renewcommand{\arraystretch}{1.3}
\begin{table}[t]    
    \caption{\textcolor{black}{Summary of Notations in Algorithms~\ref{alg:01} --~\ref{alg:03}}}
    \label{tab:notations}
    \fontsize{8pt}{8pt}\selectfont
    \centering
    \begin{tabular}{>{\centering\arraybackslash}m{1cm} m{0.8\linewidth}}
    \toprule
    \textcolor{black}{Notation} & \textcolor{black}{Explanation} \\ \hline
    \textcolor{black}{$T$} & \textcolor{black}{Output Parallel-Pair Table, each entry $T_j$ is a parallel-pairable subset of inter-links} \\
    \textcolor{black}{$T_j$} & \textcolor{black}{The $j$-th parallel-pairable group in table $T$} \\
    \textcolor{black}{$E^i$} & \textcolor{black}{The $i$-th Inter-links in the graph} \\
    \textcolor{black}{$\bar{E}^i$} & \textcolor{black}{Inter-link in complement graph, corresponding to original $E^i$} \\
    \textcolor{black}{$C(T_j)$} & \textcolor{black}{Common candidate set shared by all inter-links in group $T_j$} \\
    \textcolor{black}{$C(E^i)$} & \textcolor{black}{Set of parallel-pair candidates for inter-link $E^i$} \\
    \textcolor{black}{$\mathcal{E}$} & \textcolor{black}{Input inter-link set in Alg.~\ref{alg:02} and~\ref{alg:03} (corresponds to $R$ in Alg.~\ref{alg:01})} \\
    \textcolor{black}{$\tilde{m}$} & \textcolor{black}{Number of inter-links in input request set $\mathcal{E}$} \\
    \textcolor{black}{$m$} & \textcolor{black}{Total number of inter-links in the entire graph} \\
    \textcolor{black}{$N(E^i)$} & \textcolor{black}{Set of inter-links incompatible with $E^i$ (non-candidates)} \\
    \textcolor{black}{$N$} & \textcolor{black}{Union of all $N(E^i)$, used for conflict detection} \\    
    \bottomrule
    \end{tabular}   
\end{table}

\begin{algorithm}[t]    
    \caption{Dynamic Parallel Pairs $(CG_{n_1,\dots,n_k}, R)$}
    \label{alg:01}
    \begin{spacing}{1.05}
    \begin{algorithmic}[1]
        \Require Controlled Inter-QNet $CG_{n_1,\dots,n_k}$,
        \Statex \hspace*{1.3em} Remote Request set $R \subseteq \bar{E}$
        \Ensure Parallel-Pairs Table $T$
        \State $\bar{G}_{n_1, \dots, n_k} \gets$ \Call{Complement}{$CG_{n_1,\dots,n_k}$} 
        \State $\text{Check} \gets$ \textbf{Alg.~\ref{alg:02}} \Call{Check Parallel-Pairable}{$\bar{G}_{n_1, \dots, n_k}, R$}
        \If{$\text{Check} = \textbf{True}$} 
            \State \Return $R$ \Comment{$R$ is already parallel-pairable}
        \EndIf
        \State $T \gets \emptyset$, $j \gets 1$ \Comment{Initialize output table and group counter}
        \While{$R \neq \emptyset$}
            \State $T_j \gets \emptyset$, $C(T_j) \gets \emptyset$ \Comment{Initialize new group $T_j$ and its common candidates}
            \State $\{(\bar{E}^i, C(\bar{E}^i))\}_{i=1}^{\tilde{m}} \gets$ \textbf{Alg.~\ref{alg:03}} \Call{Parallel-Pair Candidate}{$\bar{G}_{n_1,\dots,n_k}, R$}
            \State Select $\bar{E}^i \in R$ \Comment{Pick a seed inter-link for $T_j$}
            \State $T_j \gets \{\bar{E}^i\}$, $R \gets R \setminus \{\bar{E}^i\}$
            \State $C(T_j) \gets C(\bar{E}^i) \cap R$ \Comment{Initial candidates}
            \While{$C(T_j) \neq \emptyset$}
                \State Select $\bar{E}^l \in C(T_j)$ \Comment{Add a compatible inter-link to $T_j$}
                \State $T_j \gets T_j \cup \{\bar{E}^l\}$, $R \gets R \setminus \{\bar{E}^l\}$
                \State $C(T_j) \gets C(T_j) \cap C(\bar{E}^l)$ \Comment{Update candidates}
            \EndWhile
            \State $T \gets T \cup \{T_j\}$, $j \gets j + 1$ \Comment{Finalize group $T_j$}
        \EndWhile        
        \State \Return $T$
    \end{algorithmic}
    \end{spacing}
\end{algorithm}

\begin{algorithm}[t]    
    \caption{Check Parallel-Pairable $(G_{n_1,\dots,n_k}, \mathcal{E})$}
    \label{alg:02}
    \begin{spacing}{1.05}
    \begin{algorithmic}[1]
        \Require Inter-QNet $G_{n_1,\dots,n_k}$, Inter-Links set $\mathcal{E} \subseteq E$
        \Ensure \textbf{True} if $\mathcal{E}$ is parallel-pairable, else \textbf{False}
        \State $\{(E^i, C(E^i))\}_{i=1}^{\tilde{m}} \gets$ \textbf{Alg.~\ref{alg:03}} \Call{Parallel-Pair Candidate}{$G_{n_1,\dots,n_k}, \mathcal{E}$}
        \For{$i=1$ to $\tilde{m}$}
            \State $N(E^i) \gets E \setminus C(E^i) \setminus \{E^i\}$ \Comment{Incompatible inter-links for $E^i \in \mathcal{E}$}
        \EndFor
        \State $N \gets \bigcup_{i=1}^{\tilde{m}} N(E^i)$ \Comment{Aggregate all incompatible inter-links}
        \If{$\mathcal{E} \cap N = \emptyset$}
            \State \Return \textbf{True} \Comment{No conflicts in $\mathcal{E}$}
        \Else
            \State \Return \textbf{False}
        \EndIf
    \end{algorithmic}
    \end{spacing}
\end{algorithm}

\begin{algorithm}[t]    
    \caption{Parallel-Pair Candidate $(G_{n_1,\dots,n_k}, \mathcal{E})$}
    \label{alg:03}
    \begin{spacing}{1.05}
    \begin{algorithmic}[1]
        \Require Inter-QNet $G_{n_1,\dots,n_k}$, Inter-Links set $\mathcal{E} \subseteq E$
        \Ensure Candidate list $\{(E^i, C(E^i))\}_{i=1}^{\tilde{m}}$ for $\mathcal{E}$
        \State $\{E^1, \dots, E^{\tilde{m}}\} \gets \mathcal{E}$ \Comment{Label input inter-links}
        \State $\{E^{\tilde{m}+1}, \dots, E^m\} \gets E \setminus \mathcal{E}$ \Comment{Label non-input inter-links}
        \For{$i=1$ to $\tilde{m}$}
            \State $C(E^i) \gets \emptyset$ \Comment{Initialize candidates for $E^i$}
            \For{$j=1$ to $m$}
                \If{$E^j, E^i$ satisfy \eqref{eq:def:x06} in $G_{n_1,\dots,n_k}$}
                    \State $C(E^i) \gets C(E^i) \cup \{E^j\}$ \Comment{Add a compatible inter-link}
                \EndIf
            \EndFor
        \EndFor
        \State \Return $\{(E^i, C(E^i))\}_{i=1}^{\tilde{m}}$ \Comment{Return candidates for $\mathcal{E}$}
    \end{algorithmic}
    \end{spacing}
\end{algorithm}

To address the research problem,  Lem.~\ref{lem:x01} enables immediate End-to-End entanglement establishment (Objective i). To achieve parallel communications (Objective ii), we reduce the problem to partitioning the remote request set $R$ into parallel-pairable subsets, leading to our core dynamic Parallel Pairs algorithm. \textcolor{black}{The local symbols used in the algorithm design in this section are shown in Tab.~\ref{tab:notations}; the definitions of global symbols follow those in Tab.~\ref{tab:02} above and will not be repeated.}

\textit{1) Algorithm Design:} Our solution framework integrates three polynomial-time algorithms. The core processor, Alg.~\ref{alg:01}, dynamically partitions the global request set $R$ into parallel-pairable subsets for simultaneous establishment. It is supported by Alg.~\ref{alg:03}, which generates the parallel-pair candidates for each inter-link within the entire inter-domain network, and Alg.~\ref{alg:02}, which verifies if a given set of inter-links is parallel-pairable.

Our dynamic Parallel-Pairs algorithm is described in Alg.~\ref{alg:01}. We take the
corresponding graph $CG_{n_1, \dots, n_k}$ of initial controlled Inter-QNet $\ket{CG_{n_1, \dots, n_k}}$ and Remote Request $R$ as Input, and require a dynamic parallel-pairs table $T$ as Output. The algorithm begins by transforming $CG_{n_1,\dots,n_k}$ into its complement form $\bar{G}_{n_1,\dots,n_k}$ (Lem.~\ref{lem:x01}), whereby $R$ corresponds precisely to a set of inter-links in $\bar{G}_{n_1,\dots,n_k}$. Then we check if $R$ is already parallel-pairable (Lines 2-5, Alg.~\ref{alg:01}) using   \Call{Check Parallel-Pairable}{$\bar{G}_{n_1, \dots, n_k}, R$} in Alg.~\ref{alg:02}. If true, it immediately returns $R$ and terminates. Otherwise, it begins constructing table $T$, where each entry $T_j$ is a parallel-pairable subset of inter-links, i.e., S-D pairs in $R$. The main \texttt{\textsc{while}} loop (Lines 7–19, Alg.~\ref{alg:01}) continues until all requests in $R$ are processed, selecting compatible inter-links for each group $T_j$. Specifically, at each iteration, a new group $T_j$ is initialized, along with its candidates set $C(T_j)$, representing the common set of parallel-pair candidates for all inter-links in $T_j$. Then we call  \Call{Parallel-Pair Candidate}{$\bar{G}_{n_1, \dots, n_k}, R$} in Alg.~\ref{alg:03} to compute candidates for each inter-link, i.e., S-D pair in $R$ (Line 9, Alg.~\ref{alg:01}). Specifically, a seed inter-link is autonomously selected to activate $T_j$ and make sure its candidate set $C(T_j) \subseteq R$ (Lines 11–12, Alg.~\ref{alg:01}). While $C(T_j) \neq \emptyset$, a nested \texttt{\textsc{while}} loop (Lines 13–17, Alg.~\ref{alg:01}) adds compatible inter-links from $C(T_j)$ into $T_j$, and updating $C(T_j)$ at each step. Once the nested \texttt{\textsc{while}} loop finalized, $T_j$ is added to $T$ while its inter-links are removed from $R$, and the algorithm proceeds to form the next group $T_{j+1}$. After all inter-links are processed, the algorithm returns the completed table $T$.

\begin{remark}
Our algorithm produces a dynamic output shaped by autonomously selected seed inter-links (Line 10, Alg.~\ref{alg:01}) and progressive candidate refinement (Lines 13–17, Alg.~\ref{alg:01}), enabling fully autonomous parallel-pairable group formation.
\end{remark}

Once entering \Call{Check Parallel-Pairable}{$\bar{G}_{n_1, \dots, n_k}$, $R$} functionalities, Alg.~\ref{alg:02} is executed to determine whether the input $R$ is parallel-pairable over the Complement Inter-QNet $\bar{G}_{n_1, \dots, n_k}$. Firstly, we call \Call{Parallel-Pair Candidate}{$\bar{G}_{n_1, \dots, n_k}, R$} in  Alg.~\ref{alg:03} to identify all parallel-pairs candidates for each inter-links $\bar{E}^i \in R$. Then we update the incompatible set $N(\bar{E}^i)$ (Line 2-4, Alg.~\ref{alg:02}), i.e., the set of non-parallel-pairs candidates for $\bar{E}^i$, and merge them as the union set $N$ (Line 5, Alg.~\ref{alg:02}). Determining whether $R$ is parallel-pairable simply reduces to checking whether 
$R$ and $N$ are disjoint (Lines 6–10, Alg.~\ref{alg:02}).

Once entering \Call{Parallel-Pair Candidate}{$\bar{G}_{n_1, \dots, n_k}, R$} functionalities, 
Alg.~\ref{alg:03} is executed to construct a list of parallel-pairs candidates for each inter-link $\bar{E^i}$, i.e., S-D pair in the input set $R$ over the Complement Inter-QNet $\bar{G}_{n_1, \dots, n_k}$. After the preliminary operations in lines 1-2, Alg.~\ref{alg:03} searches candidates for each inter-link $\bar{E}^i \in R$ sequentially. It examines each $\bar{E}^i$ against all other inter-links in $\bar{G}_{n_1, \dots, n_k}$ and adds those that satisfy Eq.~\eqref{eq:def:x06} (line 6, Alg.~\ref{alg:03}) to the set $C(\bar{E}^i)$. Finally, it outputs the collection of each $\bar{E}^i \in R$ and its parallel pairs candidates $C(\bar{E}^i)$.

\textit{2) Complexity Analysis:} We analyze the computational complexity of Algorithms~\ref{alg:01},~\ref{alg:02}, and~\ref{alg:03}, and establish the following desirable property: all three algorithms run in polynomial time.

\begin{theo}
    \label{theo:x01}
    For any Inter-QNet $\ket{G_{n_1, \dots, n_k} }$ with corresponding graph $G_{n_1, \dots, n_k}=(V_1, \dots, V_k, E)$, let $\mathcal{E} \subseteq E$ denote a given set of inter-links. Algorithm~\ref{alg:03} computes parallel-pair candidates for each inter-link in $\mathcal{E}$ in $O(|E|)$ time,  
    yielding a total time complexity of $O(|\mathcal{E}|\cdot|E|)$.
    \begin{IEEEproof}
        Please refer to App.~\ref{app:theo:x01,x02,x03}
    \end{IEEEproof}
\end{theo}

\begin{theo}
    \label{theo:x02}
    For any Inter-QNet $\ket{G_{n_1, \dots, n_k} }$ with corresponding graph $G_{n_1, \dots, n_k}=(V_1, \dots, V_k, E)$,  let $\mathcal{E} \subseteq E$ denote a given set of inter-links. Algorithm~\ref{alg:02} determines whether $\mathcal{E}$ is parallel-pairable with time complexity $O(|\mathcal{E}|\cdot|E|)$.
    \begin{IEEEproof}
        Please refer to App.~\ref{app:theo:x01,x02,x03}
    \end{IEEEproof}
\end{theo}

\begin{theo}
    \label{theo:x03}
   For any Controlled Inter-QNet $\ket{CG_{n_1, \dots, n_k} }$ with  associated complement Inter-QNet graph $\bar{G}_{n_1, \dots, n_k}=(V_1, \dots, V_k, \bar{E})$, let $R \subseteq \bar{E}$ denote a given set of remote connection request. Algorithm~\ref{alg:01} dynamically partitions $R$ into parallel-pairable subsets with time complexity $O(|R|^2\cdot|\bar{E}|)$.
    \begin{IEEEproof}
        Please refer to App.~\ref{app:theo:x01,x02,x03}
    \end{IEEEproof}
\end{theo}


\section{Performance Evaluation}
\label{sec:5}

\begin{figure}[t]
    \centering    \includegraphics[width=0.45\textwidth]{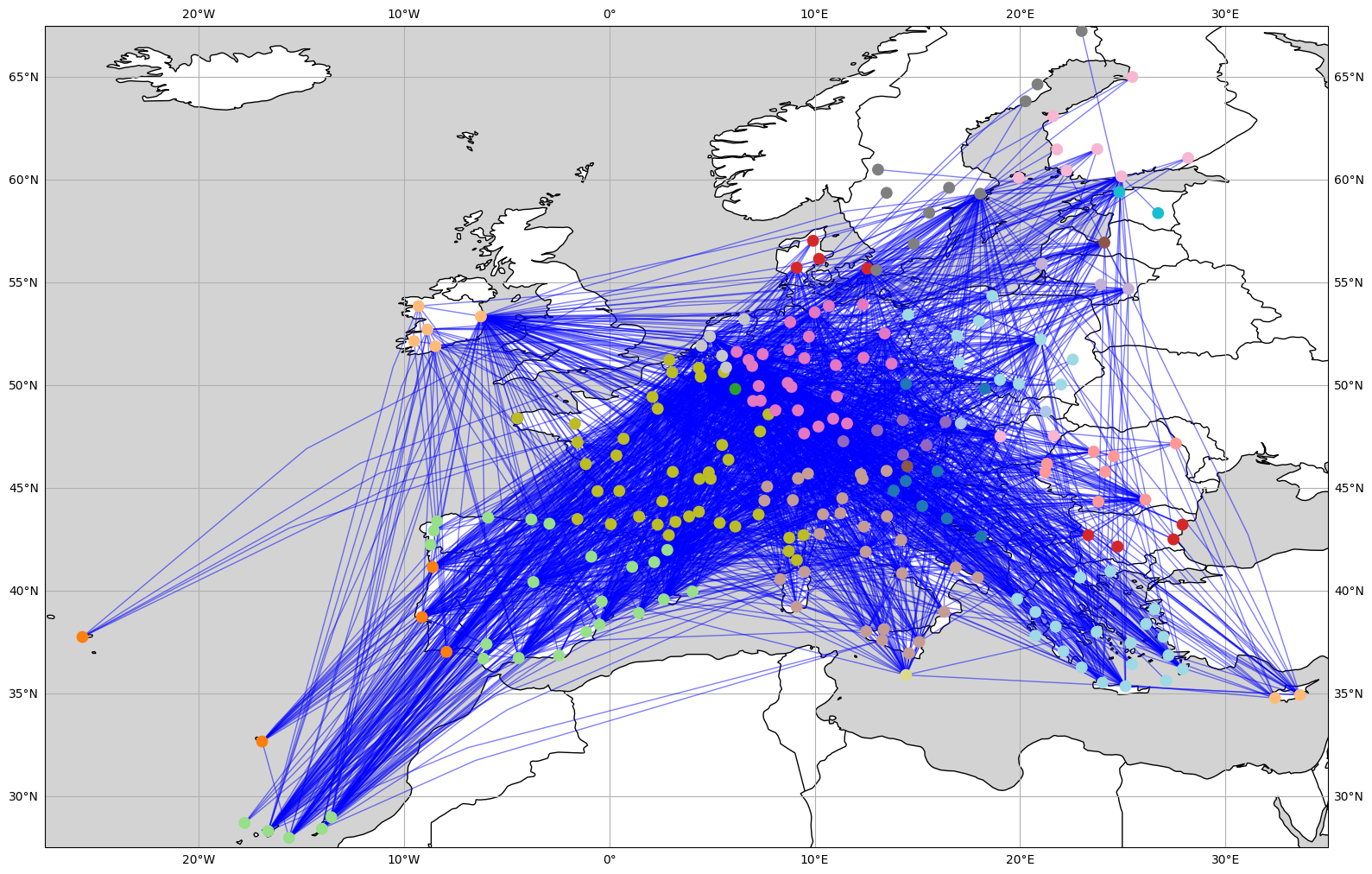}
    \caption{Schematic diagram of a sample Inter-QNet Graph utilizing international flight routes data between European Union countries, sourced from Openflight.org. Each city is modeled as a  node, with nodes from the same country forming a QNet and sharing the same color. International routes are depicted as inter-links connecting nodes across different QNets.}
    \label{fig:05}
    \vspace{0.5ex}
    \hrulefill
\end{figure}

\subsection{Methodology}
\label{sec:5.1}

\noindent  We evaluate our protocols in both synthetic and real-world inter-domain network graph datasets. For synthetic structure, we use the $k$-colorable graph generation method in the NetworkX library, such that the generated graph satisfies Inter-QNet graphical characteristics, such as the $k$ QNets distribution and inter-links structure. 
\textcolor{black}{ A density parameter $p$ is used to control the connectivity of synthetic network: after ensuring connectivity via a spanning tree, additional inter-links between nodes of different colours are added with probability $p$. Thus $p$ directly determines the expected edge density of the generated graph. In the simulations, we consider two representative regimes: $p=0.8$ (synthetic dense networks) and $p=0.2$ (synthetic sprase networks).}

For the real-world dataset, we construct an Inter-QNet network based on global aviation data. Specifically, we utilize the dataset provided by the Openflights organization (Openflight.org) \cite{OpenFlights-Data}, which comprises 3,140 nodes (representing cities) and 67,663 edges (representing airline routes). To obtain Inter-QNet network graphs of various scales, we extract subgraphs from this global dataset by uniformly sampling edges corresponding to international flight routes. In this context, each city is treated as a  node, each country is treated as a QNet, and the international airline routes between cities belonging to different countries serve as inter-links. A sampled subgraph of international airline routes between cities in European Union countries is shown in Fig.~\ref{fig:05} for illustration. Both the synthetic and real-world datasets are Inter-QNet graphs and do not contain control-nodes, so we synthetically generate a control-node system by assigning one control-node to each QNet. 

We simulate traffic by randomly choosing source-destination pairs, where the source can be any node in the network, and the destination is one of the propagated addresses. Since we focus on the inter-domain scenario, we exclude all the requests where the source and the destination are within the same QNet. \textcolor{black}{For consistency with current experimental demonstrations, we fix the entanglement-graph size to 50 nodes, a scale comparable to recent multipartite state-preparation experiments~\cite{CaoWuChe-23}. To rigorously account for structural heterogeneity across graph instances, all reported results are averaged over 1,000 independently generated graph instances.}

\subsection{\textcolor{black}{Evaluation Performance}}
\label{sec:5.2}

\subsubsection{Performance of MEC}
\label{sec:5.2.1}

\textcolor{black}{To evaluate the performance of our routing scheme, we consider two metrics that characterize key routing properties: hop-counts and throughput.}

\begin{figure}[t]
    \centering

    \subfloat[Hop in Real Data Network]{%
        \includegraphics[width=0.485\linewidth]{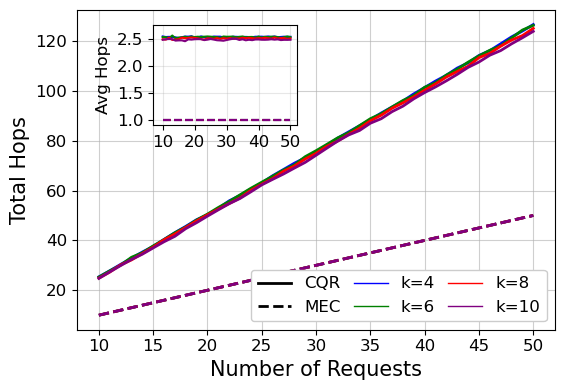}%
        \label{fig:x07.a}
    }
    \hfill
    \subfloat[Hop in Synthetic Network]{%
        \includegraphics[width=0.485\linewidth]{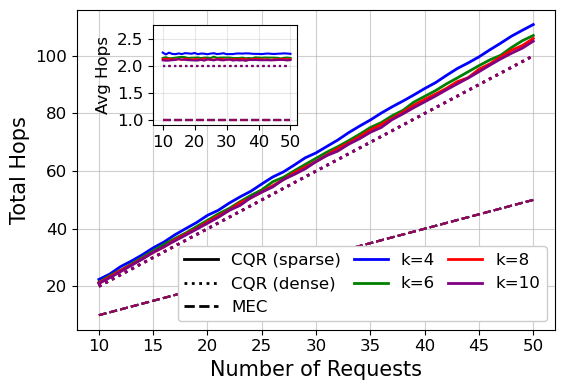}%
        \label{fig:x07.b}
    }

    \caption{Total and average hop-counts for CQR and MEC under (a) real data network and (b) synthetic network with different densities.}
    \label{fig:x07}
    \vspace{0.5ex}
    \hrulefill
\end{figure}

\setcounter{footnote}{5}

\paragraph{\textcolor{black}{{Hop-count}}}
\textcolor{black}{We use Dijkstra's algorithm -- a widely adopted shortest-path method -- as a representative CQR benchmark. To ensure a meaningful and conservative comparison, we evaluate MEC under a strict resource constraint, namely a routing-qubit footprint (RQF) of one qubit per node for both control nodes $c_i$ and ordinary nodes $v_i$. In contrast, CQR requires intermediate nodes to perform Bell-state measurements (BSM) for entanglement swapping, implying typically at least two qubits referring to two distinct EPR links. Therefore, in our simulations we relax the routing-stage qubit constraint (i.e. the routing-qubit footprint) for CQR, by allowing nodes $v_i$ to use at least two qubits and by additionally permitting control nodes $c_i$ to participate in the routing. Under these assumptions, an S--D request may be realized either through direct inter-QNet mediation (two hops via an intermediate QNet) or through control-node-assisted routing (three hops via the control nodes of the source and destination QNets). This relaxation favours shortest-path discovery for CQR, thereby providing a conservative baseline for comparison.}

\textcolor{black}{Let $\bar{h}$ denote the average hop-count per remote request on the original controlled inter-QNet graph.
As shown in Fig.~\ref{fig:x07}, the average hop-count per request under CQR ranges from $2.0$ to $2.5$, whereas under MEC it remains constant to $1$, yielding a hop-count reduction of up to 60$\%$.  As the number of requests increases, the aggregate hop-count increases accordingly and is primarily influenced by the network density rather than by the number of QNets $k$.}

\begin{figure*}[t]
\noindent
   \begin{minipage}[t]{\textwidth}
        \centering
        \resizebox{\textwidth}{!}{
            \input{Figures/timelegend}
        }
        \label{fig:x06-e1}
    \end{minipage}%
    \vspace{-0.5cm}

    \subfloat[$\lambda \geq T^M_P + T^M_R$ \label{fig:x06-a}]{
        \resizebox{0.45\textwidth}{!}{\input{Figures/time01}}
    }
    \hfill
    \subfloat[$T^M_R \leq \lambda < T^M_P + T^M_R$ \label{fig:x06-b}]{
        \resizebox{0.45\textwidth}{!}{\input{Figures/time02}}
    }

    \subfloat[$\lambda \geq T^B_P + T^B_R$ \label{fig:x06-c}]{
        \resizebox{0.45\textwidth}{!}{\input{Figures/time03}}
    }
    \hfill
    \subfloat[$\lambda < T^B_P + T^B_R$ \label{fig:x06-d}]{
        \resizebox{0.45\textwidth}{!}{\input{Figures/time04}}
    }

    \caption{\textcolor{black}{Timing diagrams illustrating possible regimes of the request interval $\lambda$. Dark-blue and light-blue blocks denote entanglement preparation time in MEC and CQR, respectively. Dark-orange and light-orange blocks denote routing time in MEC and CQR, respectively. The upper row corresponds to MEC, where resource preparation $T^{M}_{P}$ may start before the request arrival $t_i$, under the adopted proactive strategy. The bottom row corresponds to CQR, where preparation $T^{B}_{P}$ begins only after $t_i$, since EPR pairs are typically generated and distributed on demand according to specific requests. The relations between $\lambda$, the preparation time and the routing time determine how many service cycles can be completed within a window.}
    }
    \label{fig:x06}
    \vspace{0.5ex}
    \hrulefill
\end{figure*}
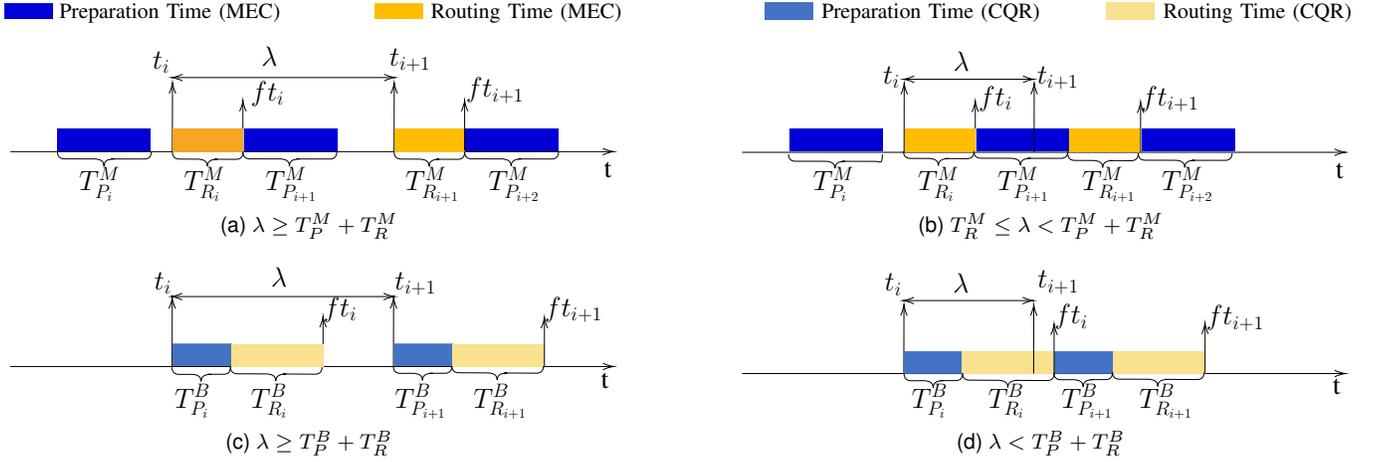

\begin{figure}[t]
    \centering

    \subfloat[Average parallel requests $\bar{r}$ completed per \textcolor{black}{cycle}]{%
        \includegraphics[width=\linewidth]{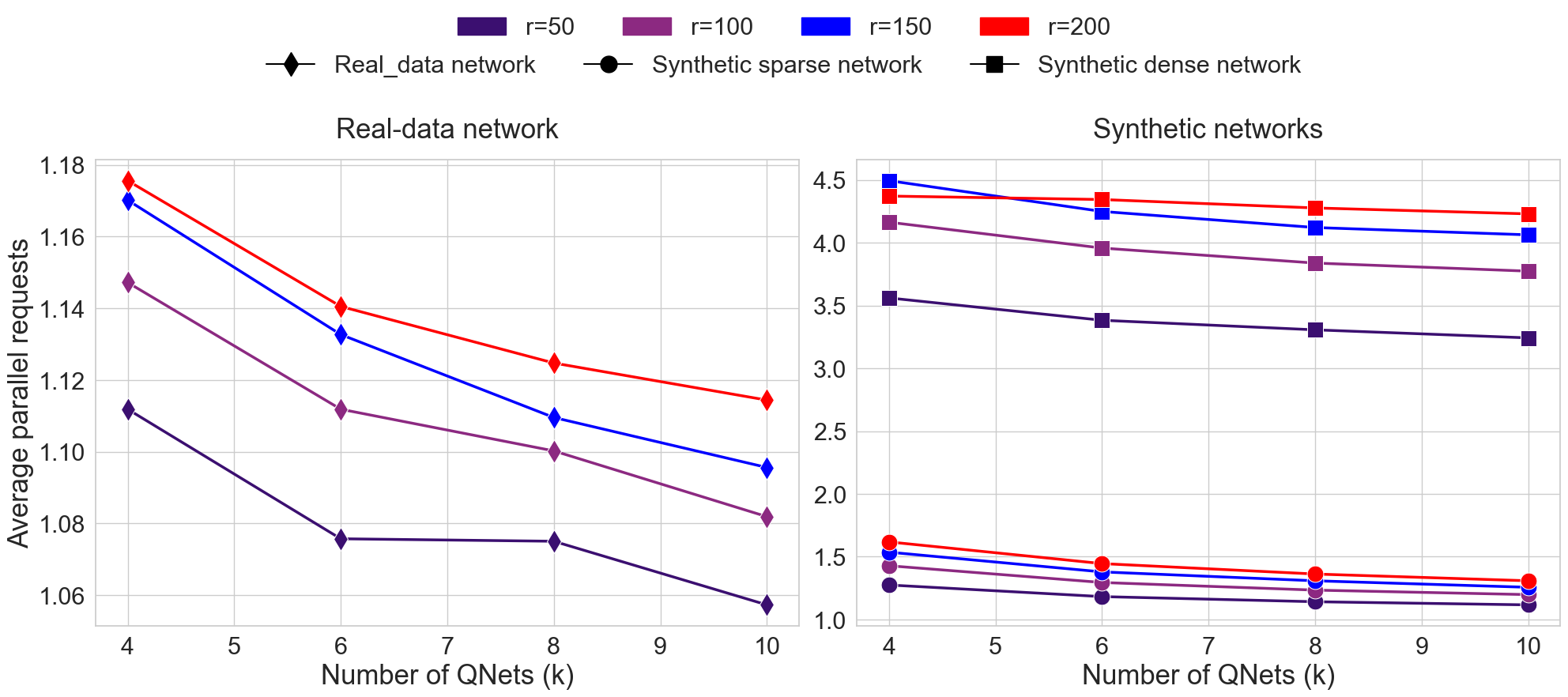}%
        \label{fig:x09.a}
    }

    \vspace{0.5em}  

    \subfloat[Average \textcolor{black}{cycles} $\rho$ required to complete varied requests]{%
        \includegraphics[width=\linewidth]{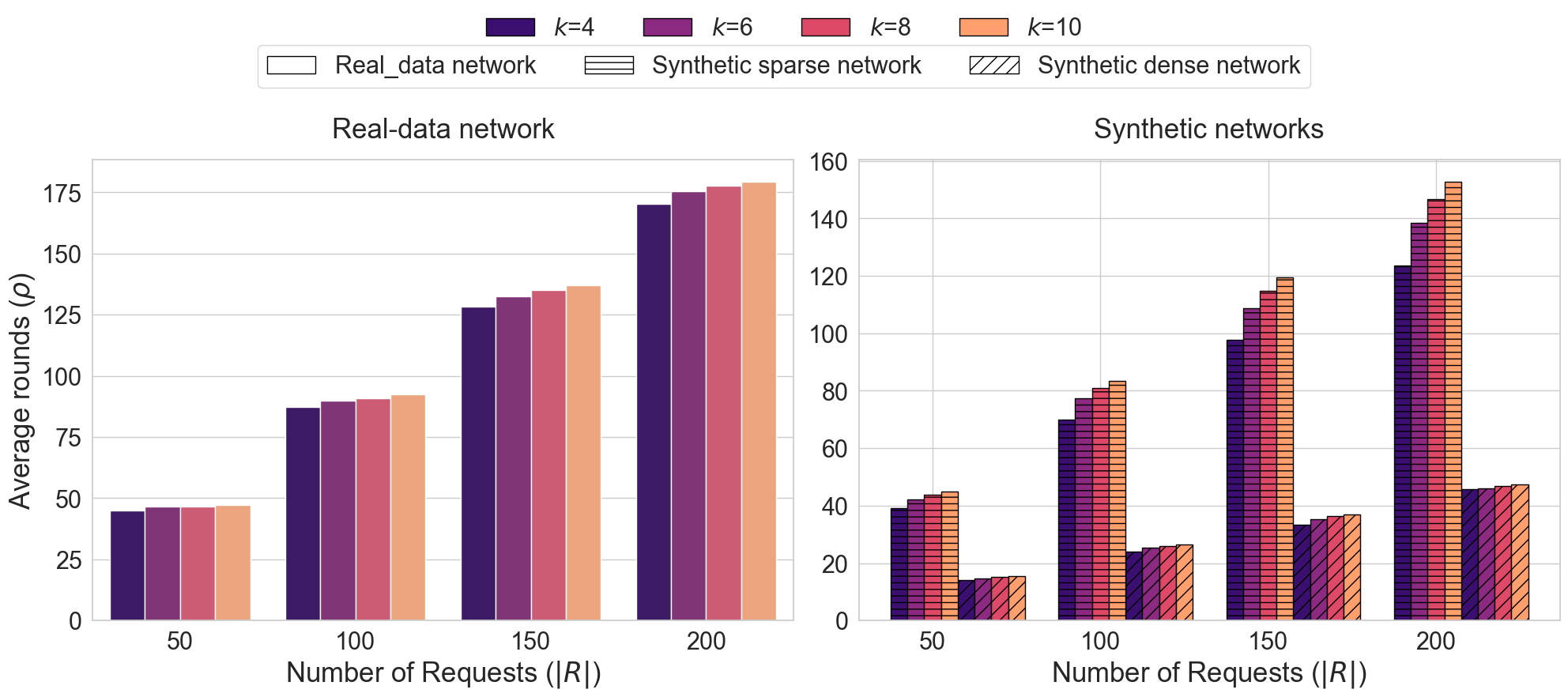}%
        \label{fig:x09.b}
    }

    \caption{Performance of the Dynamic Parallel Pairs (DP) algorithm under varying \textcolor{black}{QNet numbers} and request volumes.}
    \label{fig:x09}
    \vspace{0.5ex}
    \hrulefill
\end{figure}

\paragraph{\textcolor{black}{{Throughput}}}
\textcolor{black}{We define network throughput~\cite{IllCacMan-21} as the number of successfully established end-to-end EPR pairs within a fixed time window. Requests are released in rounds: at the beginning of round $i$, a batch of requests arrives at time $t_i$ and the inter-arrival time is $\lambda = t_{i+1}-t_i$. Over each interval of length $\lambda$,  the network executes service cycles composed of two stages: (i) resource-state preparation and (ii) routing. Accordingly, the throughput is characterized through the preparation time $T_P$ and the routing time $T_R$.}

\textcolor{black}{For MEC, we denote:}
\begin{itemize}
    \item \textcolor{black}{$T^M_{P}$: preparation time per cycle of the controlled inter-QNet resource state, including entanglement generation and distribution;}
    \item \textcolor{black}{$T^M_{R}$: routing time per cycle due to local Pauli measurements\footnote{\label{ft:TR}\textcolor{black}{MEC involves successive $X$-measurement on control-node system and $Z$-measurement on non-target nodes, while CQR $T^B_{R}$ includes successive BSMs on immediate nodes. For the sake of fairness, we assume these measurements can be performed in parallel across nodes and therefore record a single measurement-time contribution per cycle. For more details about Pauli measurement please refer to \cite{MazCalCac-25}.}}.}
\end{itemize}

\textcolor{black}{For CQR, we denote:}
\begin{itemize}
    \item \textcolor{black}{$T^B_{P}$: preparation time per cycle to generate and distribute the $\bar{h}$ EPR pairs required for a single request along a routing path with average hop-count $\bar{h}$;}
    \item \textcolor{black}{$T^B_{R}$: routing time per cycle due to BSMs\footref{ft:TR} used in entanglement swapping along the routing path.}
\end{itemize}

\textcolor{black}{The timing behaviour differs fundamentally between MEC and CQR. In MEC, each preparation cycle produces the same graph-state resource, regardless of the specific request set. Consequently, resource preparation can be performed proactively during idle periods before request arrival. When a request batch arrives at time $t_i$, the multipartite resource may already be available, and the system can immediately proceed with the routing stage. In contrast, CQR prepares entanglement on-demand for each request, since the endpoints and the routing path are only determined after the request arrives. Proactively preparing EPR resources in CQR would risk generating resources that do not match future requests, leading to both increased routing-qubit footprint requirements and entanglement waste. The resulting timing regimes are illustrated in Fig.~\ref{fig:x06}. We further note that $T_P^B$ may also include additional control-plane overhead associated with per-round path selection. For simplicity, this overhead is not explicitly illustrated in Fig.~\ref{fig:x06}, which makes the visual comparison favourable with respect to CQR. By contrast, MEC relies on a request-agnostic multipartite resource. Therefore, its distribution stage is less coupled to per-round, request-specific path selection. In any case, distribution-related overheads for both paradigms are captured by the parameters $T_P^M$ and $T_P^B$, whose relative magnitude is platform- and workload-dependent. Our framework keeps both terms explicit to support fair instantiation under different assumptions.}

\textcolor{black}{Let $\bar r$ denote the average number of requests served per cycle (empirically evaluated in Fig.~\ref{fig:x09.a}) -- i.e., the average number of end-to-end EPR pairs delivered per service cycle. Stemming from the different timing behaviour of MEC and CQR, their throughputs -- $F^M$ and $F^B$ respectively -- are given by:}
\textcolor{black}{\begin{equation}
    \label{eq:throughput1}
    F^M= 
    \begin{cases}
        \left\lfloor\dfrac{\lambda - T^M_P}{T^M_P + T^M_R} + 1 \right\rfloor \times \dfrac{\bar{r}}{\lambda}, & \text{if } \ \lambda \geq T^M_P + T^M_R \\[2ex]
        \ \dfrac{\bar{r}}{\lambda},  & \text{if } \ T^M_R \leq \lambda < T^M_P + T^M_R \\[2ex]
        \  0, & \text{if } \ \lambda < T^M_R.
    \end{cases}
\end{equation}}
\textcolor{black}{\begin{equation}
    \label{eq:throughput2}
    F^B= 
    \begin{cases}
         \left\lfloor \dfrac{\lambda}{ T^B_P + T^B_R} \right\rfloor \times \dfrac{1}{\lambda} , & \text{if } \ \lambda \geq T^B_P + T^B_R \\[2ex]
        \ 0, & \text{if } \ \lambda < T^B_P + T^B_R .
    \end{cases} 
\end{equation}}
\textcolor{black}{The detailed explanation about Eqs.~(\ref{eq:throughput1})--(\ref{eq:throughput2}) is provided in App.~\ref{app:throughput}.}

\textcolor{black}{These analytical expressions provide an analytical guideline for evaluating and comparing the performance of the two MEC and CQR routing paradigm. These formulations explicitly capture the dependence of the throughput on network characteristics (e.g., average hop-count), the adopted entanglement distribution strategy, and hardware-dependent operation times, which jointly determine the preparation time $T_P$ and the routing time $T_R$. Given the diversity of quantum hardware platforms -- including superconducting~\cite{CaoWuChe-23,SimBesUst-22}, neutral-atom~\cite{Zha-23}, and trapped-ion systems~\cite{Li-25} -- each characterized by distinct operational constraints and performance regimes, we refrain from assigning specific values to $T_P$ and $T_R$. Instead, we leave these parameters to be instantiated according to the target hardware and application scenario.}

\textcolor{black}{The aforementioned device-agnostic and application-oriented perspective allows readers to set the timing parameters under their specific technological assumptions and performance requirements, thereby avoiding platform-specific bias and ensuring broader applicability of the analysis. In Sec.~\ref{sec:6}, we further discuss the ``MEC--CQR tradeoff'': depending on the platform and provisioning strategy, MEC may incur a more resource-intensive preparation phase to enable more efficient routing phase, whereas relies on request- and path-dependent preparation, which can increase routing-stage complexity and may additionally incur per-round path-selection overhead.}

\begin{figure*}[t]
    \centering

    \subfloat[Processing load and \textcolor{black}{aggregate routing-qubit footprint} on real data sparse networks.]{%
        \includegraphics[width=\linewidth]{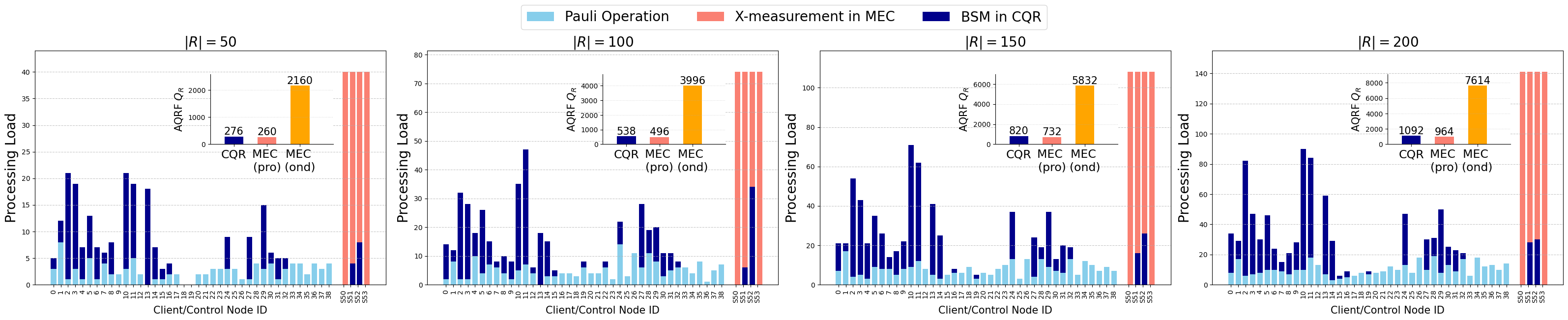}%
        \label{fig:x10.a}
    }

    \vspace{0.5em}  

    \subfloat[Processing load and \textcolor{black}{aggregate routing-qubit footprint} on synthetic dense networks.]{%
        \includegraphics[width=\linewidth]{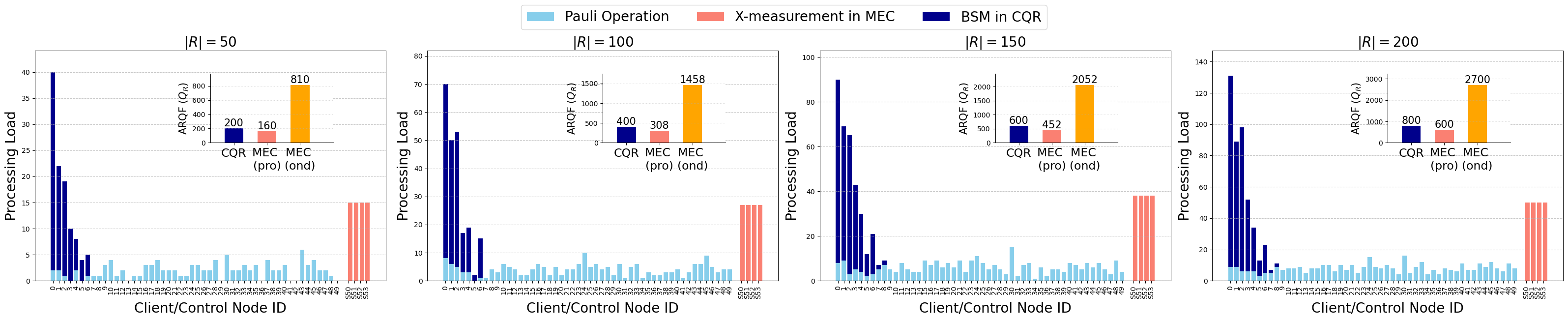}%
        \label{fig:x10.b}
    }

    \caption{Resource allocation analysis under CQR and MEC with DP algorithm, fixing $\textcolor{black}{k'}=4$ and varying $|R|$, evaluated on (a) real data sparse networks and (b) synthetic dense networks. \textcolor{black}{The main plots visualize each node aggregate routing-qubit footprint}, where nodes serving as endpoints (light blue) demand a single qubit for Pauli corrections per request, while CQR repeater nodes (dark blue) demand two qubits per Bell State Measurement. Control nodes in MEC (orange) utilize single qubits for $X$-measurements. Corresponding insets quantify \textcolor{black}{aggregate routing-qubit footprint (ARQF) $Q_R$ for the CQR, proactive MEC, and on-demand MEC approaches.}} 
    \label{fig:x10}
    \vspace{0.5ex}
    \hrulefill
\end{figure*}

\subsubsection{Performance of DP}
\label{sec:5.2.2}

We evaluate the performance of our Dynamic Parallel Pairs (DP) algorithm by measuring \textcolor{black}{parallel processing capability and aggregate routing-qubit  footprint.}

\paragraph{\textcolor{black}{{Parallel Processing Capability}}}
\textcolor{black}{To quantify DP's parallelism, we measure (i) the average number of requests served in parallel per service cycle, denoted by $\bar{r}$, and (ii) the number of service cycles -- denoted by $\rho$ -- required to complete a given request set. Let $R$ denote the set of S--D requests in a batch. Hence by definition $|R| = \rho \cdot \bar{r}$.} Specifically, we randomly select 4, 6, 8 and 10 QNets from a sample network and 
choose 50, 100, 150 and 200 source-destination pairs from these QNets. \textcolor{black}{Applying DP to schedule these requests yields Fig.~\ref{fig:x09.a} and Fig.~\ref{fig:x09.b}, which plots $\bar{r}$ and $\rho$ for different configurations, respectively. The results demonstrate that the DP algorithm consistently completes multiple requests in parallel. Overall, our DP algorithm consistently requires fewer preparation rounds than the total number of requests, with particularly significant gains in dense topology.}

\paragraph{\textcolor{black}{{Aggregate Routing-Qubit Footprint (ARQF)}}}
\textcolor{black}{To quantify routing resource overhead, we evaluate the aggregate routing-qubit footprint (ARQF), denoted as $Q_R$, i.e., the \textit{total} number of entanglement-resource qubits required network-wide during the routing phase to serve a batch of S--D requests.}\\
In CQR, which adopts BSMs \textcolor{black}{during routing phase}, each request requires one \textcolor{black}{qubit} at the source and destination, and at least two \textcolor{black}{qubits} at each intermediate node along the selected path to perform entanglement swapping. \textcolor{black}{Therefore, completing all requests by CQR requires an ARQF of}
\begin{equation}
    \label{eq:consume2}
    Q^B_{\textcolor{black}{R}}= 2|R| + 2\chi, 
\end{equation}
\textcolor{black}{where $\chi$ is total number of intermediate nodes across all requests in CQR.} 

\textcolor{black}{For MEC, ARQF depends on the adopted provisioning strategy. In the following, we consider two representative schemes, namely proactive and on-demand MEC.}

\textcolor{black}{(i) \textit{Proactive MEC:} the multipartite resource is prepared in advance and, per service cycle, one qubit is at every node of the controlled inter-QNet. Therefore, completing all requests over $\rho$ cycles requires an aggregate footprint of
\begin{equation}
    \label{eq:consume_proactive}
    Q^{M,\mathrm{pro}}_{R}= \rho\!\left(k' + \sum_{i=1}^{k}|V_i|\right),
\end{equation}
where $k'$ is the number of control nodes.}

\textcolor{black}{(ii) \textit{On-demand MEC:} if MEC adopts the same request-driven strategy as CQR (i.e., resources are prepared only after requests are known), qubits need only be prepared at the involved S--D nodes, while the control layer contributes an additional footprint scaling with $\rho$. In this case, the ARQF becomes
\begin{equation}
    \label{eq:consume1}
    Q^{M,\mathrm{ond}}_{R}= 2|R| + \rho k'.
\end{equation}}

\textcolor{black}{To further illustrate how DP mitigates load concentration across nodes/control nodes and to quantify ARQF}, in Fig.~\ref{fig:x10} we fix \textcolor{black}{$k'=4$ and vary $|R|$, reporting} the processing load on each node/control node and \textcolor{black}{the corresponding ARQF under CQR and MEC. The results show that CQR distributes requests relatively evenly in sparse networks but can concentrate load on specific intermediate nodes in dense graphs, whereas MEC concentrates routing-stage operations at control nodes while leaving ordinary nodes largely unaffected.} \textcolor{black}{Compared to CQR, when DP is integrated with proactive MEC, the ARQF remains relatively high for the network-wide pre-preparation requirement. In contrast, under the on-demand MEC with DP, the ARQF can be reduced in both sparse and dense networks, achieving performance comparable to or better than CQR.}

\section{Conclusion \textcolor{black}{and Discussion}}
\label{sec:6}

\noindent In this paper, we introduced a Multipartite Entanglement Complementation (MEC) strategy that enables direct entanglement establishment between arbitrary remote source-destination \textcolor{black}{pairs beyond conventional pathfinding. We also developed an accompanying orchestration algorithm to autonomously select and parallelize multiple entanglement requests.} \textcolor{black}{Simulation results indicate that MEC can significantly reduce hop-count and routing-qubit footprint, while preserving scalability across multi-domain quantum networks. Overall, the proposed framework highlights a promising routing-control paradigm that improves resource efficiency and operational autonomy compared to conventional routing approaches}.

\textcolor{black}{
In the following, we discuss some key implications of the proposed framework and outline promising directions for future research.}

\paragraph{\textcolor{black}{Multipartite Entanglement Generation}} 
\textcolor{black}{Although multipartite resource-state generation is not the central focus of the present work, it is important to clarify our resource accounting. The controlled inter-QNet states underlying MEC are characterized in this work in terms of their \textit{logical-qubit} count, corresponding to the number of vertices of the target network state. The actual number of \textit{physical qubits} required for preparing such states is implementation-dependent, as it depends on the adopted generation technique the capabilities of the underlying hardware. In practice, preparing multipartite resource states may entail additional physical-qubit overhead associated with preparation procedures such as micro-cluster generation~\cite{Nie-04} and fusion-based operations~\cite{ThoRusMor-22,ThoRusMpr-24,ThoRusMor-22-2}. Recent work~\cite{MazMigIll-25} on flexible qubit allocation formalizes this distinction, showing that the effective ``granularity'' of each logical system can be adapted to meet architectural and resilience requirements.}

\textcolor{black}{We emphasize that such preparation overhead is not unique to MEC. CQR schemes also face non-negligible physical-qubit consumption due to quantum error correction techniques, such as purification procedures. In this sense, the physical-qubit overhead associated with entanglement generation is a general and vital challenge in quantum networking.}

\textcolor{black}{Finally, ongoing theoretical~\cite{Nie-04,PatPanEng-22,RubMasYel-25} and experimental~\cite{ThoRusMor-22,ThoRusMpr-24,ThoRusMor-22-2} progresses in multipartite entanglement generation suggest that, while resource overhead is unavoidable, scalable preparation of highly multipartite entangled states is becoming increasingly feasible.}

\paragraph{\textcolor{black}{Control Nodes and Entanglement Distribution}} 
\textcolor{black}{The cost of establishing entanglement links generally depends on factors such as physical distance and on the adopted entanglement-generation/distribution strategy supported by the underlying hardware. These considerations relate to a broader co-design problem between ``clustering formation'' and ``control-layer entanglement provisioning'', as recently discussed in quantum-internet architectural frameworks~\cite{rfc-cac,CalCac-26}. Importantly, our MEC framework operates at the routing abstraction level and assumes the availability of entanglement resources, providing the required logical connectivity. Consequently, our analysis does not impose explicit constraints on the physical distances of the links realizing such resources. Within this abstraction, it is useful to distinguish between two types of connections: (i) links between a control node and the nodes it serves, and (ii) entangled links among control nodes.}

\textcolor{black}{\textit{(i) Links between a control node and its served nodes.} These links can be interpreted as \textit{intra-domain} entanglement connections determined by the clustering strategy used to partition the network. In practice, clustering policies typically incorporate physical proximity within an entangling-cost metric. Accordingly, most node–control links are expected to remain short-range and relatively low-cost. Moreover, the distinction between ``nodes'' and ``control nodes'' should be viewed primarily as an architectural abstraction. For instance, a control node can be \textit{elected within each QNet}~\cite{CheIllCac-25} according to an optimization objective, for example, minimizing entanglement cost or improving connectivity. Such variations, while outside the scope of this work, do not affect the routing-level analysis presented here and would only require adapting the node-set definitions introduced in Sec.~\ref{sec:2}.}

\textcolor{black}{\textit{(ii) Entangled links among control nodes.} Establishing entangled links at the control layer is an architectural design choice that can follow principles adopted in modern quantum-internet designs for inter-domain connectivity~\cite{rfc-cac,CalCac-26}. Hierarchical architectures typically rely on a \textit{sparse} set of long-range (high-cost) entangled links among higher-tier nodes, while maintaining predominantly local, low-cost connectivity. As a result, long-distance entanglement is sustained by a limited subset of higher-tier nodes and supported by appropriate mechanisms for long-distance entanglement distribution~\cite{CalCac-26}, rather than requiring a fully connected entanglement topology across all control nodes. This is also consistent with the architectural paradigm of separating control and data planes~\cite{rfc-cac,CalCac-26}. Therefore, the number and placement of control-layer entangled links can be deliberately constrained and optimized at the architectural level. The detailed design and evaluation of such provisioning strategies are beyond the scope of this work, but can be addressed using established architecture-level methodologies~\cite{rfc-cac,CalCac-26}.}

\textcolor{black}{Overall, our analysis is conducted at the routing-level abstraction, under the assumption that the necessary entanglement resources are readily available for use. Concrete implementation choices -- including clustering strategies and entanglement provisioning policies -- can follow existing architectural frameworks~\cite{rfc-cac,CalCac-26}. A comprehensive optimization and quantitative assessment of the associated distance/entangling-cost tradeoffs remains an important direction for future work.}

\paragraph{\textcolor{black}{MEC-CQR Tradeoff}}
\textcolor{black}{The tradeoff between MEC and CQR is rooted in how each paradigm allocates cost between the \textit{preparation} phase and the \textit{routing} phase, which jointly determine throughput via the timing parameters $(T_P,T_R)$ in Eqs.~\eqref{eq:throughput1} and \eqref{eq:throughput2}. Since $T_P$ and $T_R$ depend on network characteristics (e.g., average hop count), the adopted entanglement distribution strategy, and hardware-dependent operation times, the relative performance of MEC and CQR is inherently platform- and workload-dependent, as already discussed in Sec.~\ref{sec:5}. }
\textcolor{black}{On the preparation side, MEC relies on generating and distributing a multipartite resource state whose logical connectivity is request-agnostic. Depending on the platform and provisioning strategy, this can translate into a higher preparation effort. In contrast, CQR prepares entanglement in a request- and path-dependent manner, which can reduce the need for proactive, network-wide resource provisioning but ties preparation to the specific request set and selected paths. As noted in Sec.~\ref{sec:5}, this request-driven nature can also introduce additional overhead for per-round path selection, which may be accounted for within the derived throughput under the chosen timing model.}

\textcolor{black}{On the routing side, the asymmetry reverses: once a multipartite resource is available, MEC serves requests through local single-qubit Pauli-measurements and feed-forward, thereby reducing both operational complexity and routing-qubit footprint at each node. Conversely, CQR relies on path-dependent entanglement swapping along request-specific paths, which increases routing-stage overhead as the network grows.}

\textcolor{black}{Overall, MEC can be viewed as a ``front-loaded'' approach: depending on the platform and provisioning strategy, it may incur a more resource-intensive preparation phase for enabling a more efficient routing phase based primarily on local single-qubit operations. In contrast, CQR follows a more ``just-in-time'' approach based on request- and path-dependent preparation, which can increase routing-stage complexity. Consequently, the relative advantage of MEC versus CQR depends on the interplay between device capabilities, provisioning assumptions, and network scale: MEC may be more attractive when stable pre-distribution and orchestration are feasible, whereas CQR may remain preferable in more resource-constrained or modest settings.}

\paragraph{\textcolor{black}{Operational Errors and Practical Considerations}}
\textcolor{black}{Operational errors are an important practical factor when implementing graph-state–based routing protocols, and in particular when executing the local operations involved in graph-state transitions. While a quantitative analysis, under hardware-specific noise models, is beyond the scope of this work, we highlight here the main implications for MEC and CQR and discuss how such effects can be incorporated in future studies.} 

\textcolor{black}{Current experimental platforms already achieve high gate fidelities. Single- and two-qubit gate fidelities in superconducting and trapped-ion systems typically exceed 99\%~\cite{CaoWuChe-23,Goo-23,Ion-23}, while neutral-atom platforms report two-qubit fidelities around 99.5\%~\cite{EveDolMar-23}. More recent trapped-ion experiments report single-qubit error rates as low as $1.5\times10^{-5}$~\cite{MaLiuPen-23}. These results indicate that, although non-zero, operational errors in multipartite state manipulation are increasingly well controlled in state-of-the-art implementations.}

\textcolor{black}{Regarding a comparison with CQR, an important architectural difference affecting error accumulation arises during the routing stage. In CQR, establishing end-to-end entanglement typically requires a sequence of entanglement-swapping operations along a multi-hop path. As the path length increases, each additional swapping step introduces further two-qubit gates and measurements, leading to accumulated errors. In contrast, MEC leverages a pre-distributed multipartite resource and realizes end-to-end connectivity primarily through local single-qubit measurements (e.g., $X$-measurements) and feed-forward corrections, without requiring a chain of path-dependent swapping operations during routing. Thus, while both frameworks remain subject to noise degradation, MEC may reduce the number of multi-qubit operations in the routing phase that scale with the hop-count, thereby mitigating error accumulation compared to conventional multi-hop swapping–based routing. In this context, it is important to note that graph states belong to the stabilizer-state family, whose error behaviour under Pauli noise can be systematically analysed using the noise stabilizer formalism~\cite{MorDur-23}. This provides a useful foundation for developing explicit error-propagation models and noise-aware routing strategies in future work.
}

\section*{Acknowledgment}
\noindent The authors thank Rajiuddin Sk and Dick Maryopi for proofreading the earlier version of this manuscript.

\appendices

\section{Proof of Lemma~\ref{lem:x01}}
\label{app:lem:x01}
Assume $\ket{ CG_{n_1,\dots,n_k} }$ is defined as in Eq.~\eqref{eq:def:x03}. We aim to prove that $\ket{ \bar{G}_{n_1,\dots,n_k} }$ can be obtained from it. For clarity, we adhere to the following notation, which refers specifically to the initial associated graph $CG{n_1,\dots,n_k}$:

For a node $v_i$ in $CG{n_1,\dots,n_k}$, let $N^i$ and $\overline{N}^{i}$ denote its sets of adjacent and non-adjacent (remote) nodes in initial $CG_{n_1,\dots,n_k}$, respectively. 
More precisely, we define $N^i_l \eqdef N^i \cap V_l$ and $\overline{N}^i_l \eqdef \overline{N}^i \cap V_l$ as the adjacent and remote nodes of $v_i$ within the QNet $V_l$.

Furthermore, for any two QNets $V_a$ and $V_b$, the set of initial inter-links $E(V_a, V_b)$ and the set of complement inter-links $\bar{E}(V_a, V_b)$ between them are defined as:
\begin{align}    
    \label{eq:app:lem:x01.00}
    & E(V_a, V_b) \eqdef   \bigcup_{v_{i} \in V_a} \big( v_{i}, (v_{j} \in {N}^{i}) \cap V_b \big) \\
    \label{eq:app:lem:x01.01}
    & \bar{E}(V_a, V_b) \eqdef   \bigcup_{v_{i} \in V_a} \big( v_{i}, (v_{j} \in \overline{N}^{i}) \cap V_b \big)
\end{align}

The proof constructively follows by performing Pauli-$x$ measurement 
on control-nodes $c_1, \dots, c_{k'}$ in order. 
\begin{itemize}
    \item[i)] Pauli-$x$ measurement on the control-nodes $c_1$ with the arbitrary neighbor nodes $v_{1.i}$ as $k_0$\cite{CheIllCac-25}. It is equivalent to perform the sequence of graph operations $\tau_{v_{1.i}} \left( \tau_{c_1}\big(\tau_{v_{1.i}}(CG_{n_1,\dots,n_k})\big)-c_1 \right)$. For clarify, let's use $v_{1.1}$ as $k_0$.  This yields to $G^1$ as:
        \begin{align}
            \label{eq:app:lem:x01.1}
            & G^1  = \tau_{v_{1.1}} \left( \tau_{c_1}\big(\tau_{v_{1.1}}(CG_{n_1,\dots,n_k})\big)-c_1 \right)= \nonumber \\
            & =( \big( \bigcup_{i=1}^k V_i \big) \cup \big( \bigcup_{j=2}^{k'} c_j \big), \bigcup_{ \substack{ 2 \leq j, l \leq k' \\ j \neq l} }(c_j,c_l)   \cup \nonumber \\
            & \cup \big( \bigcup_{j=2}^{k'}(c_j, v_{1.1}) \big) \cup \big( \bigcup_{ \substack{ 2 \leq j \leq k' \\ 2 \leq l \leq k \\ j \neq l} }(c_j, N_l^{1.1}) \big) \cup \nonumber \\
            & \cup \big( \bigcup_{ \substack{ j=2 } }^k(c_j, \bar{N}_j^{1.1}) \big)  \cup \big( \bigcup_{i=2}^{n_1} (v_{1.1}, v_{1.i}) \big) \cup \nonumber \\
            & \cup \big( \bigcup_{ \substack{ 2 \leq i \leq n_1 \\ 2 \leq j \leq k } } \bar{E}(v_{1.i}, N_j^{1.1}) \big) \cup \big( \bigcup_{ \substack{ 2\leq i,j \leq k \\ i\ne j } }  E(V_i, V_j) \big) \big)
        \end{align}        
    \item[ii)] Pauli-$x$ measurement on the control-nodes $c_2$ by choosing again $v_{1.1}$ as $k_0$ (which belongs now to $N(c_2)$ as a consequence of the first Pauli-$x$ measurement). It is equivalent to perform the sequence of graph operations $\tau_{v_{1.1}} \left( \tau_{c_2}\big(\tau_{v_{1.1}}(G^1)\big)-c_2 \right)$. This yields to $G^2$:
        \begin{align}
            \label{eq:app:lem:x01.2}
            & G^2  = \tau_{v_{1.1}} \left( \tau_{c_2}\big(\tau_{v_{1.1}}(G^1)\big)-c_2 \right)= \nonumber \\
            & \big(  \big( \bigcup_{i=1}^k V_i \big) \cup \big( \bigcup_{j=3}^{k'} c_j \big), \bigcup_{ \substack{ 3 \leq j,l \leq k' \\ j \neq l} } (c_j,c_l) \cup  \nonumber \\
            &\cup \big( \bigcup_{j=3}^{k'}\bigcup_{l=1}^2(c_j, V_l) \big)  \cup \big( \bigcup_{j=3}^{k}(c_j, V_j)  \big) \cup   \nonumber \\ 
            &\cup \big( \bigcup_{ \substack{ 1 \leq i,j \leq 2 \\ i \neq j} } \bar{E}(V_i,V_j) \big) \cup \big( \bigcup_{ \substack{ 1 \leq i \leq k,\\ 3 \leq j \leq k, \\ i \ne j } } E(V_i, V_j) \big) \big) 
        \end{align}
\end{itemize}
We can observe that, compared to the original graph $CG_{n_1, \dots,n_k}$, the graph $G^2$ results from transformations: (1) remove  $c_1, c_2$ and corresponding edges; (2) connect all  nodes in $V_1, V_2$ to each remaining control node $c_i$ for $i \geq 3$; (3) complement Inter-QNet links between $V_1$ and $V_2$.

We define a family of graphs $G^d$, derived from the original Inter-QNet graph $CG_{n_1, \dots,n_k}$, where $d$ $(d\leq k'-2)$ denotes the number of measured control nodes. The construction of $G^d$ follows these rules:
\begin{enumerate}
    \item Control-node System: 
    \begin{enumerate}[label=.]
        \item The $d$ measured control nodes $c_1, \dots, c_d$ are removed.
        \item The remaining $(k'-d)$ control nodes $c_{d+1}, \dots, c_{k'}$ are fully connected.
    \end{enumerate}    
    \item Node-Control Interface:
    \begin{enumerate}[label=.]
        \item Each node in $V_i$ for $d<i\leq k$ connects only to its  control node $c_i$, as in the original graph.
        \item Each node in $V_j$ for $j\leq d$ connects to all remaing control nodes $c_i$ where $d<i \leq k'$.
    \end{enumerate}     
    \item Inter-QNet Links: The subgraph induced by $V_1 \cup \dots \cup V_d$ has complement inter-links, i.e., all inter-links among these node sets are complemented relative to the original graph.
\end{enumerate}
Based on the above construction, we define the corresponding graph state $\ket{G^d}$ associated with the graph $G^d$ as:
\begin{align}
    \label{eq:app:lem:x01.3}
    & G^d  = \big(  \big( \bigcup_{i=1}^k V_i \big) \cup \big( \bigcup_{j=d+1}^{k'} c_j \big), \bigcup_{ \substack{ d+1 \leq j, l \leq k' \\ j \neq l} }(c_j,c_l) \cup  \nonumber \\
    &\cup \big( \bigcup_{j=d+1}^{k'}\bigcup_{l=1}^d(c_j, V_l) \big)  \cup \big( \bigcup_{j=d+1}^k(c_j, V_j)  \big) \cup   \nonumber \\ 
    &\cup \big( \bigcup_{ \substack{ 1 \leq i,j \leq d \\ i \neq j} } \bar{E}(V_i,V_j) \big) \cup \big( \bigcup_{ \substack{ 1 \leq i \leq k\\ d+1 \leq j \leq k \\ i \ne j } } E (V_i, V_j) \big) \big)
\end{align}

After two times of Pauli-$x$ measurement on $c_{d+1}, c_{d+2}$ sequentially, with both $k_0=v_{1.1}$, it collapses to a graph state $\ket{G^{d+2}}$, with corresponding graph $G^{d+2}$ as:
\begin{align}
    \label{eq:app:lem:x01.4}
    & G^{d+2}  = \big(  \big( \bigcup_{i=1}^k V_i \big) \cup \big( \bigcup_{j=d+3}^{k'} c_j \big), \bigcup_{ \substack{ d+3 \leq j, l \leq k' \\ j \neq l} }(c_j,c_l) \cup  \nonumber \\
    &\cup \big( \bigcup_{j=d+3}^{k'}\bigcup_{l=1}^{d+2}(c_j, V_l) \big)  \cup \big( \bigcup_{j=d+3}^k(c_j, V_j)  \big) \cup   \nonumber \\ 
    &\cup \big( \bigcup_{ \substack{ 1 \leq i, \,j \leq d+2 \\ i \neq j} } \bar{E}(V_i,V_j) \big) \cup \big( \bigcup_{ \substack{ 1 \leq i \leq k,\\ d+3 \leq j \leq k \\ i \ne j } } E(V_i, V_j) \big) \big)
\end{align}

Since the derivation process is analogous to steps i), ii), we omit the details for brevity. Starting from $G^2$, successive Pauli-$x$ measurements on control nodes $c_3, \dots, c_{k'-2}$ all with same $k_0=v_{1.1}$, yield $G^{k'-2}$ (i.e., $d=k'-2$) as described in Eq.~\eqref{eq:app:lem:x01.3}. Proceeding with two additional Pauli-$x$ measurements on control nodes $c_{k'-1}$ and $c_{k'}$, again with $k_0=v_{1.1}$, results in a graph state $\ket{G^{k'}}$, corresponding to the graph $G^{k'}$:
\begin{equation}   
    \label{eq:app:lem:x01.8}
    G^{k'}  = \big(  \bigcup_{i=1}^k V_i , \bigcup_{ \substack{ 2 \leq i,j \leq k \\ i\neq j } } \bar{E}(V_i, V_j) \big) = \bar{G}_{n_1,\dots,n_k}
\end{equation}

Thus the proof follows.

\section{Proof of Theorem~\ref{theo:x01},~\ref{theo:x02} and~\ref{theo:x03}}
\label{app:theo:x01,x02,x03}

For Theorem~\ref{theo:x01}, Alg.~\ref{alg:03} performs parallel-pair searching for $|\mathcal{E}|$ inter-links sequentially. In each round, it checks $E^i \in \mathcal{E}$ against all $|E|$ inter-links in $G_{n_1, \dots, n_k}$ based on the parallel-pair condition in~\eqref{eq:def:x04} (Line 3-10). This yields an $O(|\mathcal{E}|\cdot|E|)$ time complexity.

For Theorem~\ref{theo:x02}, Alg.~\ref{alg:02}  executes a $O(|\mathcal{E}|)$ loop operation while dominantly relying on Algorithm~\ref{alg:03}, resulting in the same $O(|\mathcal{E}|\cdot|E|)$ complexity.

For Theorem~\ref{theo:x03}, Algorithm~\ref{alg:01} dominates with its $O(|R|)$ main \texttt{\textsc{while}} loop iterations, each performing $O(|R|\cdot|\bar{E}|)$ operations via Alg.~\ref{alg:03}, yielding total $O(|R|^2\cdot|\bar{E}|)$ time.

Hence, all the proposed algorithms run in polynomial time.

\section{\textcolor{black}{Details on Throughput in MEC and CQR}}
\label{app:throughput}

\paragraph{\textcolor{black}{{Throughput of MEC in \eqref{eq:throughput1}}}}

\begin{itemize}
    \item \textcolor{black}{Case (a1): $\lambda \geq T^M_P + T^M_R$.} \\ 
    \textcolor{black}{This regime is illustrated in Fig.~\ref{fig:x06-a}. Under the proactive MEC setting, the multipartite resource can be prepared in advance. Hence, upon request arrival at $t_i$, the system immediately executes the routing stage and the completion time satisfies $ft_i=t_i+T^M_R$.}  \textcolor{black}{If $\lambda$ is sufficiently large, MEC can complete the routing stage and also re-prepare the resource for subsequent cycles within the same interval. Therefore, the number of MEC cycles $n^M$ that can be executed within $\lambda$ is:
    \begin{equation}
    \label{eq:FMa}
    n^M = \lfloor \frac{\lambda -T^M_P }{ T^M_P + T^M_R}  +1 \rfloor.
    \end{equation}}
    \textcolor{black}{Since each cycle serves on average $\bar r$ requests, \eqref{eq:throughput1} follows.}
    \item \textcolor{black}{Case (a2): $T^M_R \leq \lambda < T^M_P + T^M_R$.} \\
    \textcolor{black}{This regime is illustrated in Fig.~\ref{fig:x06-b}. In this case, only one MEC cycle can be completed within the interval, but there is not enough time to fully re-prepare the resource within the same window. As a result, across consecutive windows the system alternates between windows in which a prepared resource is available (one completion) and windows in which resource preparation occupies the interval (no completion). The average number of completed cycles over windows of length $\lambda$ converges to a constant value, which yields the throughput expression in \eqref{eq:throughput1}.}

    \item \textcolor{black}{Case (a3): $\lambda < T^M_R$.} \\
    \textcolor{black}{In this regime, no request can be completed within the window of lenght $\lambda$, i.e., $F^M = 0$. Thus \eqref{eq:throughput1} is obtained again.}
\end{itemize}

\paragraph{\textcolor{black}{Throughput of CQR in \eqref{eq:throughput2}}}

\begin{itemize}
    \item \textcolor{black}{Case (b1): $\lambda \geq T^B_P + T^B_R$.} \\
    \textcolor{black}{ This regime is illustrated in Fig.~\ref{fig:x06-c}. In CQR, preparation starts only after request arrival. If $\lambda$ is sufficiently large, the number of CQR cycles $n^B$ that can be executed within an interval of length $\lambda$ is:
    \begin{equation}
        \label{eq:FBc}
        n^B = \lfloor \frac{\lambda }{ T^B_P + T^B_R} \rfloor.
    \end{equation} }
    \textcolor{black}{Since in our CQR timing model, each cycle serves a single request, \eqref{eq:throughput2} follows.}
    \item \textcolor{black}{Case (b2): $\lambda < T^B_P + T^B_R$.} \\
    \textcolor{black}{This regime is illustrated in Fig.~\ref{fig:x06-d}. No request can be completed within the window, i.e., $F^B = 0$, which corresponds to the second case of \eqref{eq:throughput2}.}
\end{itemize}

\bibliographystyle{IEEEtran}
\bibliography{biblio.bib}

\end{document}

%% file: Figures/timelegend.tex
\tikzset{every picture/.style={line width=0.75pt}} 

\begin{tikzpicture}[x=0.75pt,y=0.75pt,yscale=-1,xscale=1]

\draw  [draw opacity=0][fill={rgb, 255:red, 0; green, 0; blue, 214 }  ,fill opacity=1 ] (2.14,1701.67) -- (62.33,1701.67) -- (62.33,1720.67) -- (2.14,1720.67) -- cycle ;

\draw  [draw opacity=0][fill={rgb, 255:red, 255; green, 189; blue, 0 }  ,fill opacity=1 ] (460.47,1701.67) -- (519.81,1701.67) -- (519.81,1720.33) -- (460.47,1720.33) -- cycle ;
\draw  [draw opacity=0][fill={rgb, 255:red, 68; green, 114; blue, 196 }  ,fill opacity=1 ] (941.97,1699.83) -- (1001.17,1699.83) -- (1001.17,1721.17) -- (941.97,1721.17) -- cycle ;

\draw  [draw opacity=0][fill={rgb, 255:red, 249; green, 225; blue, 141 }  ,fill opacity=1 ] (1363.97,1700.83) -- (1423.17,1700.83) -- (1423.17,1722.17) -- (1363.97,1722.17) -- cycle ;

\draw (69.33,1700) node [anchor=north west][inner sep=0.75pt]  [font=\huge] [align=left] {Preparation Time (MEC)};
\draw (528,1700) node [anchor=north west][inner sep=0.75pt]  [font=\huge] [align=left] {Routing Time (MEC)};
\draw (1010.5,1700) node [anchor=north west][inner sep=0.75pt]  [font=\huge] [align=left] {Preparation Time (CQR)};
\draw (1432.5,1700) node [anchor=north west][inner sep=0.75pt]  [font=\huge] [align=left] {Routing Time (CQR)};

\end{tikzpicture}

%% file: Figures/time01.tex
\tikzset{every picture/.style={line width=0.75pt}} 

\begin{tikzpicture}[x=0.75pt,y=0.75pt,yscale=-1,xscale=1]

\draw    (60.29,141.33) -- (657.06,141.33) ;
\draw [shift={(659.06,141.33)}, rotate = 180] [color={rgb, 255:red, 0; green, 0; blue, 0 }  ][line width=0.75]    (10.93,-3.29) .. controls (6.95,-1.4) and (3.31,-0.3) .. (0,0) .. controls (3.31,0.3) and (6.95,1.4) .. (10.93,3.29)   ;
\draw    (441,141.33) -- (441.13,74.76) ;
\draw [shift={(441.13,72.76)}, rotate = 90.11] [color={rgb, 255:red, 0; green, 0; blue, 0 }  ][line width=0.75]    (10.93,-3.29) .. controls (6.95,-1.4) and (3.31,-0.3) .. (0,0) .. controls (3.31,0.3) and (6.95,1.4) .. (10.93,3.29)   ;
\draw    (221,141.33) -- (221.13,74.76) ;
\draw [shift={(221.13,72.76)}, rotate = 90.11] [color={rgb, 255:red, 0; green, 0; blue, 0 }  ][line width=0.75]    (10.93,-3.29) .. controls (6.95,-1.4) and (3.31,-0.3) .. (0,0) .. controls (3.31,0.3) and (6.95,1.4) .. (10.93,3.29)   ;
\draw   (107,141.5) .. controls (106.99,146.17) and (109.32,148.5) .. (113.99,148.51) -- (143.15,148.56) .. controls (149.82,148.57) and (153.14,150.91) .. (153.13,155.58) .. controls (153.14,150.91) and (156.48,148.59) .. (163.15,148.6)(160.15,148.59) -- (193.09,148.65) .. controls (197.76,148.66) and (200.1,146.33) .. (200.11,141.66) ;
\draw    (222.13,67.76) -- (440.13,67.76) ;
\draw [shift={(442.13,67.76)}, rotate = 180] [color={rgb, 255:red, 0; green, 0; blue, 0 }  ][line width=0.75]    (10.93,-3.29) .. controls (6.95,-1.4) and (3.31,-0.3) .. (0,0) .. controls (3.31,0.3) and (6.95,1.4) .. (10.93,3.29)   ;
\draw [shift={(220.13,67.76)}, rotate = 0] [color={rgb, 255:red, 0; green, 0; blue, 0 }  ][line width=0.75]    (10.93,-3.29) .. controls (6.95,-1.4) and (3.31,-0.3) .. (0,0) .. controls (3.31,0.3) and (6.95,1.4) .. (10.93,3.29)   ;
\draw    (291,140.33) -- (291,92.33) ;
\draw [shift={(291,90.33)}, rotate = 90] [color={rgb, 255:red, 0; green, 0; blue, 0 }  ][line width=0.75]    (10.93,-3.29) .. controls (6.95,-1.4) and (3.31,-0.3) .. (0,0) .. controls (3.31,0.3) and (6.95,1.4) .. (10.93,3.29)   ;
\draw   (441.11,141.66) .. controls (441.08,146.33) and (443.4,148.67) .. (448.07,148.69) -- (465.78,148.78) .. controls (472.45,148.81) and (475.77,151.16) .. (475.75,155.83) .. controls (475.77,151.16) and (479.11,148.85) .. (485.78,148.88)(482.78,148.86) -- (503.97,148.97) .. controls (508.64,148.99) and (510.98,146.67) .. (511,142) ;
\draw  [draw opacity=0][fill={rgb, 255:red, 245; green, 166; blue, 35 }  ,fill opacity=1 ] (221,118.5) -- (291,118.5) -- (291,141.33) -- (221,141.33) -- cycle ;
\draw  [draw opacity=0][fill={rgb, 255:red, 255; green, 189; blue, 0 }  ,fill opacity=1 ] (441,118.5) -- (511,118.5) -- (511,141.33) -- (441,141.33) -- cycle ;

\draw  [draw opacity=0][fill={rgb, 255:red, 0; green, 0; blue, 214 }  ,fill opacity=1 ] (292,118.5) -- (384.48,118.5) -- (384.48,140.33) -- (292,140.33) -- cycle ;
\draw  [draw opacity=0][fill={rgb, 255:red, 0; green, 0; blue, 214 }  ,fill opacity=1 ] (107,118.5) -- (199.48,118.5) -- (199.48,140.33) -- (107,140.33) -- cycle ;
\draw  [draw opacity=0][fill={rgb, 255:red, 0; green, 0; blue, 214 }  ,fill opacity=1 ] (512,118.5) -- (604.48,118.5) -- (604.48,140.33) -- (512,140.33) -- cycle ;

\draw    (511,140.33) -- (511,92.33) ;
\draw [shift={(511,90.33)}, rotate = 90] [color={rgb, 255:red, 0; green, 0; blue, 0 }  ][line width=0.75]    (10.93,-3.29) .. controls (6.95,-1.4) and (3.31,-0.3) .. (0,0) .. controls (3.31,0.3) and (6.95,1.4) .. (10.93,3.29)   ;

\draw   (291,141.5) .. controls (290.99,146.17) and (293.32,148.5) .. (297.99,148.51) -- (327.15,148.56) .. controls (333.82,148.57) and (337.14,150.91) .. (337.13,155.58) .. controls (337.14,150.91) and (340.48,148.59) .. (347.15,148.6)(344.15,148.59) -- (377.09,148.65) .. controls (381.76,148.66) and (384.1,146.33) .. (384.11,141.66) ;
\draw   (511,141.5) .. controls (510.99,146.17) and (513.32,148.5) .. (517.99,148.51) -- (547.15,148.56) .. controls (553.82,148.57) and (557.14,150.91) .. (557.13,155.58) .. controls (557.14,150.91) and (560.48,148.59) .. (567.15,148.6)(564.15,148.59) -- (597.09,148.65) .. controls (601.76,148.66) and (604.1,146.33) .. (604.11,141.66) ;
\draw   (221.11,141.66) .. controls (221.08,146.33) and (223.4,148.67) .. (228.07,148.69) -- (245.78,148.78) .. controls (252.45,148.81) and (255.77,151.16) .. (255.75,155.83) .. controls (255.77,151.16) and (259.11,148.85) .. (265.78,148.88)(262.78,148.86) -- (283.97,148.97) .. controls (288.64,148.99) and (290.98,146.67) .. (291,142) ;

\draw (648,150.33) node [anchor=north west][inner sep=0.75pt]  [font=\huge] [align=left] {t};
\draw (308,35) node [anchor=north west][inner sep=0.75pt]  [font=\huge]  {$\lambda$};
\draw (200,40.33) node [anchor=north west][inner sep=0.75pt]  [font=\huge]  {$t_{i}$};
\draw (435,37.33) node [anchor=north west][inner sep=0.75pt]  [font=\huge]  {$t_{i+1}$};
\draw (297.5,70) node [anchor=north west][inner sep=0.75pt]  [font=\huge]  {$ft_{i}$};
\draw (127,155) node [anchor=north west][inner sep=0.75pt]  [font=\huge]  {$T_{P_{i}}^{M}$};
\draw (312,155) node [anchor=north west][inner sep=0.75pt]  [font=\huge]  {$T_{P_{i+1}}^{M}$};
\draw (533,155) node [anchor=north west][inner sep=0.75pt]  [font=\huge]  {$T_{P_{i+2}}^{M}$};
\draw (231,155) node [anchor=north west][inner sep=0.75pt]  [font=\huge]  {$T_{R_{i}}^{M}$};
\draw (451,155) node [anchor=north west][inner sep=0.75pt]  [font=\huge]  {$T_{R_{i+1}}^{M}$};
\draw (511,66.33) node [anchor=north west][inner sep=0.75pt]  [font=\huge]  {$ft_{i+1}$};

\end{tikzpicture}

%% file: Figures/time02.tex
\tikzset{every picture/.style={line width=0.75pt}} 

\begin{tikzpicture}[x=0.75pt,y=0.75pt,yscale=-1,xscale=1]

\draw    (61.67,382) -- (658.44,382) ;
\draw [shift={(660.44,382)}, rotate = 180] [color={rgb, 255:red, 0; green, 0; blue, 0 }  ][line width=0.75]    (10.93,-3.29) .. controls (6.95,-1.4) and (3.31,-0.3) .. (0,0) .. controls (3.31,0.3) and (6.95,1.4) .. (10.93,3.29)   ;
\draw    (222.33,382) -- (222.46,315.42) ;
\draw [shift={(222.46,313.42)}, rotate = 90.11] [color={rgb, 255:red, 0; green, 0; blue, 0 }  ][line width=0.75]    (10.93,-3.29) .. controls (6.95,-1.4) and (3.31,-0.3) .. (0,0) .. controls (3.31,0.3) and (6.95,1.4) .. (10.93,3.29)   ;
\draw    (223.46,309.42) -- (349.33,309.42) ;
\draw [shift={(351.33,309.42)}, rotate = 180] [color={rgb, 255:red, 0; green, 0; blue, 0 }  ][line width=0.75]    (10.93,-3.29) .. controls (6.95,-1.4) and (3.31,-0.3) .. (0,0) .. controls (3.31,0.3) and (6.95,1.4) .. (10.93,3.29)   ;
\draw [shift={(221.46,309.42)}, rotate = 0] [color={rgb, 255:red, 0; green, 0; blue, 0 }  ][line width=0.75]    (10.93,-3.29) .. controls (6.95,-1.4) and (3.31,-0.3) .. (0,0) .. controls (3.31,0.3) and (6.95,1.4) .. (10.93,3.29)   ;
\draw    (293,381) -- (293,333) ;
\draw [shift={(293,331)}, rotate = 90] [color={rgb, 255:red, 0; green, 0; blue, 0 }  ][line width=0.75]    (10.93,-3.29) .. controls (6.95,-1.4) and (3.31,-0.3) .. (0,0) .. controls (3.31,0.3) and (6.95,1.4) .. (10.93,3.29)   ;

\draw  [draw opacity=0][fill={rgb, 255:red, 0; green, 0; blue, 214 }  ,fill opacity=1 ] (294,359.17) -- (386.48,359.17) -- (386.48,381) -- (294,381) -- cycle ;
\draw  [draw opacity=0][fill={rgb, 255:red, 0; green, 0; blue, 214 }  ,fill opacity=1 ] (109,359.17) -- (201.48,359.17) -- (201.48,381) -- (109,381) -- cycle ;
\draw  [draw opacity=0][fill={rgb, 255:red, 0; green, 0; blue, 214 }  ,fill opacity=1 ] (458,359.17) -- (550.48,359.17) -- (550.48,381) -- (458,381) -- cycle ;

\draw  [draw opacity=0][fill={rgb, 255:red, 255; green, 189; blue, 0 }  ,fill opacity=1 ] (386.48,359.17) -- (456.48,359.17) -- (456.48,382) -- (386.48,382) -- cycle ;
\draw  [draw opacity=0][fill={rgb, 255:red, 255; green, 189; blue, 0 }  ,fill opacity=1 ] (223,359.17) -- (293,359.17) -- (293,382) -- (223,382) -- cycle ;

\draw    (457,382) -- (457,334) ;
\draw [shift={(457,332)}, rotate = 90] [color={rgb, 255:red, 0; green, 0; blue, 0 }  ][line width=0.75]    (10.93,-3.29) .. controls (6.95,-1.4) and (3.31,-0.3) .. (0,0) .. controls (3.31,0.3) and (6.95,1.4) .. (10.93,3.29)   ;
\draw    (351.33,382) -- (351.46,315.42) ;
\draw [shift={(351.46,313.42)}, rotate = 90.11] [color={rgb, 255:red, 0; green, 0; blue, 0 }  ][line width=0.75]    (10.93,-3.29) .. controls (6.95,-1.4) and (3.31,-0.3) .. (0,0) .. controls (3.31,0.3) and (6.95,1.4) .. (10.93,3.29)   ;
\draw   (291.67,380.83) .. controls (291.66,385.5) and (293.98,387.84) .. (298.65,387.85) -- (327.81,387.9) .. controls (334.48,387.91) and (337.81,390.24) .. (337.8,394.91) .. controls (337.81,390.24) and (341.14,387.92) .. (347.81,387.93)(344.81,387.93) -- (377.76,387.98) .. controls (382.43,387.99) and (384.76,385.67) .. (384.77,381) ;
\draw   (221.77,380.99) .. controls (221.75,385.66) and (224.07,388) .. (228.74,388.03) -- (246.45,388.11) .. controls (253.12,388.14) and (256.44,390.49) .. (256.41,395.16) .. controls (256.44,390.49) and (259.78,388.18) .. (266.45,388.21)(263.45,388.2) -- (284.63,388.3) .. controls (289.3,388.33) and (291.64,386.01) .. (291.67,381.34) ;
\draw   (107.67,383.83) .. controls (107.66,388.5) and (109.98,390.84) .. (114.65,390.85) -- (143.81,390.9) .. controls (150.48,390.91) and (153.81,393.24) .. (153.8,397.91) .. controls (153.81,393.24) and (157.14,390.92) .. (163.81,390.93)(160.81,390.93) -- (193.76,390.98) .. controls (198.43,390.99) and (200.76,395) .. (200.77,384) ;
\draw   (454.67,382.83) .. controls (454.66,387.5) and (456.98,389.84) .. (461.65,389.85) -- (490.81,389.9) .. controls (497.48,389.91) and (500.81,392.24) .. (500.8,396.91) .. controls (500.81,392.24) and (504.14,389.92) .. (510.81,389.93)(507.81,389.93) -- (540.76,389.98) .. controls (545.43,389.99) and (547.76,387.67) .. (547.77,383) ;
\draw   (384.77,382.99) .. controls (384.75,387.66) and (387.07,390) .. (391.74,390.03) -- (409.45,390.11) .. controls (416.12,390.14) and (419.44,392.49) .. (419.41,397.16) .. controls (419.44,392.49) and (422.78,390.18) .. (429.45,390.21)(426.45,390.2) -- (447.63,390.3) .. controls (452.3,390.33) and (454.64,388.01) .. (454.67,383.34) ;

\draw (649.33,391) node [anchor=north west][inner sep=0.75pt]  [font=\huge] [align=left] {t};
\draw (270,280) node [anchor=north west][inner sep=0.75pt]  [font=\huge]  {$\lambda$};
\draw (200,295) node [anchor=north west][inner sep=0.75pt]  [font=\huge]  {$t_{i}$};
\draw (353,295) node [anchor=north west][inner sep=0.75pt]  [font=\huge]  {$t_{i+1}$};
\draw (295,315) node [anchor=north west][inner sep=0.75pt]  [font=\huge]  {$ft_{i}$};
\draw (456,315) node [anchor=north west][inner sep=0.75pt]  [font=\huge]  {$ft_{i+1}$};
\draw (130.67,395) node [anchor=north west][inner sep=0.75pt]  [font=\huge]  {$T_{P_{i}}^{M}$};
\draw (313.67,395) node [anchor=north west][inner sep=0.75pt]  [font=\huge]  {$T_{P_{i+1}}^{M}$};
\draw (234.67,395) node [anchor=north west][inner sep=0.75pt]  [font=\huge]  {$T_{R_{i}}^{M}$};
\draw (476.67,395) node [anchor=north west][inner sep=0.75pt]  [font=\huge]  {$T_{P_{i+2}}^{M}$};
\draw (397.67,395) node [anchor=north west][inner sep=0.75pt]  [font=\huge]  {$T_{R_{i+1}}^{M}$};

\end{tikzpicture}

%% file: Figures/time03.tex
\tikzset{every picture/.style={line width=0.75pt}} 

\begin{tikzpicture}[x=0.75pt,y=0.75pt,yscale=-1,xscale=1]

\draw    (58.29,1055.33) -- (656,1055.33) ;
\draw [shift={(658,1055.33)}, rotate = 180] [color={rgb, 255:red, 0; green, 0; blue, 0 }  ][line width=0.75]    (10.93,-3.29) .. controls (6.95,-1.4) and (3.31,-0.3) .. (0,0) .. controls (3.31,0.3) and (6.95,1.4) .. (10.93,3.29)   ;
\draw    (439,1055.33) -- (439.13,988.76) ;
\draw [shift={(439.13,986.76)}, rotate = 90.11] [color={rgb, 255:red, 0; green, 0; blue, 0 }  ][line width=0.75]    (10.93,-3.29) .. controls (6.95,-1.4) and (3.31,-0.3) .. (0,0) .. controls (3.31,0.3) and (6.95,1.4) .. (10.93,3.29)   ;
\draw    (219,1055.33) -- (219.13,988.76) ;
\draw [shift={(219.13,986.76)}, rotate = 90.11] [color={rgb, 255:red, 0; green, 0; blue, 0 }  ][line width=0.75]    (10.93,-3.29) .. controls (6.95,-1.4) and (3.31,-0.3) .. (0,0) .. controls (3.31,0.3) and (6.95,1.4) .. (10.93,3.29)   ;
\draw   (218.9,1056.16) .. controls (218.91,1060.83) and (221.25,1063.15) .. (225.92,1063.14) -- (237.81,1063.11) .. controls (244.48,1063.09) and (247.82,1065.41) .. (247.83,1070.08) .. controls (247.82,1065.41) and (251.14,1063.07) .. (257.81,1063.05)(254.81,1063.06) -- (270.02,1063.02) .. controls (274.69,1063.01) and (277.01,1060.67) .. (277,1056) ;
\draw    (220.13,985.76) -- (256,985.76) -- (438.13,985.76) ;
\draw [shift={(440.13,985.76)}, rotate = 180] [color={rgb, 255:red, 0; green, 0; blue, 0 }  ][line width=0.75]    (10.93,-3.29) .. controls (6.95,-1.4) and (3.31,-0.3) .. (0,0) .. controls (3.31,0.3) and (6.95,1.4) .. (10.93,3.29)   ;
\draw [shift={(218.13,985.76)}, rotate = 0] [color={rgb, 255:red, 0; green, 0; blue, 0 }  ][line width=0.75]    (10.93,-3.29) .. controls (6.95,-1.4) and (3.31,-0.3) .. (0,0) .. controls (3.31,0.3) and (6.95,1.4) .. (10.93,3.29)   ;
\draw  [color={rgb, 255:red, 68; green, 114; blue, 196 }  ,draw opacity=1 ][fill={rgb, 255:red, 68; green, 114; blue, 196 }  ,fill opacity=1 ] (220,1033) -- (278,1033) -- (278,1054.33) -- (220,1054.33) -- cycle ;

\draw    (369,1055.33) -- (369,1007.33) ;
\draw [shift={(369,1005.33)}, rotate = 90] [color={rgb, 255:red, 0; green, 0; blue, 0 }  ][line width=0.75]    (10.93,-3.29) .. controls (6.95,-1.4) and (3.31,-0.3) .. (0,0) .. controls (3.31,0.3) and (6.95,1.4) .. (10.93,3.29)   ;
\draw   (277,1055.5) .. controls (276.97,1060.17) and (279.29,1062.51) .. (283.96,1062.54) -- (312.08,1062.69) .. controls (318.75,1062.72) and (322.07,1065.07) .. (322.04,1069.74) .. controls (322.07,1065.07) and (325.41,1062.76) .. (332.08,1062.8)(329.08,1062.78) -- (360.96,1062.95) .. controls (365.63,1062.98) and (367.97,1060.66) .. (368,1055.99) ;
\draw  [color={rgb, 255:red, 68; green, 114; blue, 196 }  ,draw opacity=1 ][fill={rgb, 255:red, 68; green, 114; blue, 196 }  ,fill opacity=1 ] (440,1033) -- (498,1033) -- (498,1054.33) -- (440,1054.33) -- cycle ;

\draw  [draw opacity=0][fill={rgb, 255:red, 249; green, 225; blue, 141 }  ,fill opacity=1 ] (277.62,1033) -- (369,1033) -- (369,1055.33) -- (277.62,1055.33) -- cycle ;
\draw  [draw opacity=0][fill={rgb, 255:red, 249; green, 225; blue, 141 }  ,fill opacity=1 ] (497.62,1033) -- (589,1033) -- (589,1055.33) -- (497.62,1055.33) -- cycle ;

\draw    (589,1056.33) -- (589,1008.33) ;
\draw [shift={(589,1006.33)}, rotate = 90] [color={rgb, 255:red, 0; green, 0; blue, 0 }  ][line width=0.75]    (10.93,-3.29) .. controls (6.95,-1.4) and (3.31,-0.3) .. (0,0) .. controls (3.31,0.3) and (6.95,1.4) .. (10.93,3.29)   ;
\draw   (438.9,1055.16) .. controls (438.91,1059.83) and (441.25,1062.15) .. (445.92,1062.14) -- (457.81,1062.11) .. controls (464.48,1062.09) and (467.82,1064.41) .. (467.83,1069.08) .. controls (467.82,1064.41) and (471.14,1062.07) .. (477.81,1062.05)(474.81,1062.06) -- (490.02,1062.02) .. controls (494.69,1062.01) and (497.01,1059.67) .. (497,1055) ;
\draw   (497,1054.5) .. controls (496.97,1059.17) and (499.29,1061.51) .. (503.96,1061.54) -- (532.08,1061.69) .. controls (538.75,1061.72) and (542.07,1064.07) .. (542.04,1068.74) .. controls (542.07,1064.07) and (545.41,1061.76) .. (552.08,1061.8)(549.08,1061.78) -- (580.96,1061.95) .. controls (585.63,1061.98) and (587.97,1059.66) .. (588,1054.99) ;

\draw (644,1060) node [anchor=north west][inner sep=0.75pt]  [font=\huge] [align=left] {t};
\draw (200,960) node [anchor=north west][inner sep=0.75pt]  [font=\huge]  {$t_{i}$};
\draw (439,960) node [anchor=north west][inner sep=0.75pt]  [font=\huge]  {$t_{i+1}$};
\draw (220,1071.33) node [anchor=north west][inner sep=0.75pt]  [font=\huge]  {$T_{P_{i}}^{B}$};
\draw (296,1071.33) node [anchor=north west][inner sep=0.75pt]  [font=\huge]  {$T_{R_{i}}^{B}$};
\draw (439,1071.33) node [anchor=north west][inner sep=0.75pt]  [font=\huge]  {$T_{P_{i+1}}^{B}$};
\draw (517,1071.33) node [anchor=north west][inner sep=0.75pt]  [font=\huge]  {$T_{R_{i+1}}^{B}$};
\draw (316,955) node [anchor=north west][inner sep=0.75pt]  [font=\huge]  {$\lambda$};
\draw (368,984.33) node [anchor=north west][inner sep=0.75pt]  [font=\huge]  {$ft_{i}$};
\draw (588,984.33) node [anchor=north west][inner sep=0.75pt]  [font=\huge]  {$ft_{i+1}$};

\end{tikzpicture}

%% file: Figures/time04.tex
\tikzset{every picture/.style={line width=0.75pt}} 

\begin{tikzpicture}[x=0.75pt,y=0.75pt,yscale=-1,xscale=1]

\draw    (59.67,1271) -- (657.38,1271) ;
\draw [shift={(659.38,1271)}, rotate = 180] [color={rgb, 255:red, 0; green, 0; blue, 0 }  ][line width=0.75]    (10.93,-3.29) .. controls (6.95,-1.4) and (3.31,-0.3) .. (0,0) .. controls (3.31,0.3) and (6.95,1.4) .. (10.93,3.29)   ;
\draw    (220.33,1271) -- (220.46,1204.42) ;
\draw [shift={(220.46,1202.42)}, rotate = 90.11] [color={rgb, 255:red, 0; green, 0; blue, 0 }  ][line width=0.75]    (10.93,-3.29) .. controls (6.95,-1.4) and (3.31,-0.3) .. (0,0) .. controls (3.31,0.3) and (6.95,1.4) .. (10.93,3.29)   ;
\draw    (221.46,1198.42) -- (347.33,1198.42) ;
\draw [shift={(349.33,1198.42)}, rotate = 180] [color={rgb, 255:red, 0; green, 0; blue, 0 }  ][line width=0.75]    (10.93,-3.29) .. controls (6.95,-1.4) and (3.31,-0.3) .. (0,0) .. controls (3.31,0.3) and (6.95,1.4) .. (10.93,3.29)   ;
\draw [shift={(219.46,1198.42)}, rotate = 0] [color={rgb, 255:red, 0; green, 0; blue, 0 }  ][line width=0.75]    (10.93,-3.29) .. controls (6.95,-1.4) and (3.31,-0.3) .. (0,0) .. controls (3.31,0.3) and (6.95,1.4) .. (10.93,3.29)   ;
\draw  [color={rgb, 255:red, 68; green, 114; blue, 196 }  ,draw opacity=1 ][fill={rgb, 255:red, 68; green, 114; blue, 196 }  ,fill opacity=1 ] (220.67,1249.33) -- (278.67,1249.33) -- (278.67,1270.67) -- (220.67,1270.67) -- cycle ;
\draw  [color={rgb, 255:red, 68; green, 114; blue, 196 }  ,draw opacity=1 ][fill={rgb, 255:red, 68; green, 114; blue, 196 }  ,fill opacity=1 ] (370.33,1249.33) -- (428.33,1249.33) -- (428.33,1270.67) -- (370.33,1270.67) -- cycle ;
\draw  [draw opacity=0][fill={rgb, 255:red, 249; green, 225; blue, 141 }  ,fill opacity=1 ] (427.95,1249.33) -- (519.33,1249.33) -- (519.33,1271.67) -- (427.95,1271.67) -- cycle ;
\draw  [draw opacity=0][fill={rgb, 255:red, 249; green, 225; blue, 141 }  ,fill opacity=1 ] (278.28,1249.33) -- (369.67,1249.33) -- (369.67,1271.67) -- (278.28,1271.67) -- cycle ;

\draw    (369.67,1271.67) -- (369.67,1223.67) ;
\draw [shift={(369.67,1221.67)}, rotate = 90] [color={rgb, 255:red, 0; green, 0; blue, 0 }  ][line width=0.75]    (10.93,-3.29) .. controls (6.95,-1.4) and (3.31,-0.3) .. (0,0) .. controls (3.31,0.3) and (6.95,1.4) .. (10.93,3.29)   ;
\draw    (349.33,1271) -- (349.46,1204.42) ;
\draw [shift={(349.46,1202.42)}, rotate = 90.11] [color={rgb, 255:red, 0; green, 0; blue, 0 }  ][line width=0.75]    (10.93,-3.29) .. controls (6.95,-1.4) and (3.31,-0.3) .. (0,0) .. controls (3.31,0.3) and (6.95,1.4) .. (10.93,3.29)   ;

\draw    (519.67,1270.67) -- (519.67,1222.67) ;
\draw [shift={(519.67,1220.67)}, rotate = 90] [color={rgb, 255:red, 0; green, 0; blue, 0 }  ][line width=0.75]    (10.93,-3.29) .. controls (6.95,-1.4) and (3.31,-0.3) .. (0,0) .. controls (3.31,0.3) and (6.95,1.4) .. (10.93,3.29)   ;
\draw   (220.57,1271.16) .. controls (220.58,1275.83) and (222.92,1278.15) .. (227.59,1278.14) -- (239.48,1278.11) .. controls (246.15,1278.09) and (249.49,1280.41) .. (249.5,1285.08) .. controls (249.49,1280.41) and (252.81,1278.07) .. (259.48,1278.05)(256.48,1278.06) -- (271.69,1278.02) .. controls (276.36,1278.01) and (278.68,1275.67) .. (278.67,1271) ;
\draw   (278.67,1270.5) .. controls (278.64,1275.17) and (280.96,1277.51) .. (285.63,1277.54) -- (313.75,1277.69) .. controls (320.42,1277.72) and (323.74,1280.07) .. (323.71,1284.74) .. controls (323.74,1280.07) and (327.08,1277.76) .. (333.75,1277.8)(330.75,1277.78) -- (362.63,1277.95) .. controls (367.3,1277.98) and (369.64,1275.66) .. (369.67,1270.99) ;
\draw   (369.57,1270.16) .. controls (369.58,1274.83) and (371.92,1277.15) .. (376.59,1277.14) -- (388.48,1277.11) .. controls (395.15,1277.09) and (398.49,1279.41) .. (398.5,1284.08) .. controls (398.49,1279.41) and (401.81,1277.07) .. (408.48,1277.05)(405.48,1277.06) -- (420.69,1277.02) .. controls (425.36,1277.01) and (427.68,1274.67) .. (427.67,1270) ;
\draw   (427.67,1269.5) .. controls (427.64,1274.17) and (429.96,1276.51) .. (434.63,1276.54) -- (462.75,1276.69) .. controls (469.42,1276.72) and (472.74,1279.07) .. (472.71,1283.74) .. controls (472.74,1279.07) and (476.08,1276.76) .. (482.75,1276.8)(479.75,1276.78) -- (511.63,1276.95) .. controls (516.3,1276.98) and (518.64,1274.66) .. (518.67,1269.99) ;

\draw (645.33,1275) node [anchor=north west][inner sep=0.75pt]  [font=\huge] [align=left] {t};
\draw (268,1170) node [anchor=north west][inner sep=0.75pt]  [font=\huge]  {$\lambda$};
\draw (200,1170.33) node [anchor=north west][inner sep=0.75pt]  [font=\huge]  {$t_{i}$};
\draw (352,1167.33) node [anchor=north west][inner sep=0.75pt]  [font=\huge]  {$t_{i+1}$};
\draw (368.67,1200.67) node [anchor=north west][inner sep=0.75pt]  [font=\huge]  {$ft_{i}$};
\draw (519.67,1201.67) node [anchor=north west][inner sep=0.75pt]  [font=\huge]  {$ft_{i+1}$};
\draw (226.67,1280) node [anchor=north west][inner sep=0.75pt]  [font=\huge]  {$T_{P_{i}}^{B}$};
\draw (302.67,1280) node [anchor=north west][inner sep=0.75pt]  [font=\huge]  {$T_{R_{i}}^{B}$};
\draw (375.67,1280) node [anchor=north west][inner sep=0.75pt]  [font=\huge]  {$T_{P_{i+1}}^{B}$};
\draw (452.67,1280) node [anchor=north west][inner sep=0.75pt]  [font=\huge]  {$T_{R_{i+1}}^{B}$};

\end{tikzpicture}